\documentclass[10pt,a4paper,onecolumn]{IEEEtran}
\usepackage{blindtext, graphicx}
\usepackage{tabularx}
\usepackage{pgf,tikz,pgfplots}
\usepackage[ruled,vlined]{algorithm2e}
\usepackage{hyperref}
\usepackage{centernot}
\pgfplotsset{compat=1.14}
\usepackage{mathrsfs}
\usetikzlibrary{arrows,shapes,arrows}
\hypersetup{
 colorlinks,
 linkcolor={blue!100!black},
 citecolor={blue!100!black},
 urlcolor={blue!80!black}
}

\usepackage{amsmath,amssymb,amsfonts,amsthm,graphicx,bm,bbm,dsfont}
\newtheorem{theorem}{Theorem}

\newtheorem{proposition}{Proposition}

\newtheorem{lemma}{Lemma}

\DeclareSymbolFont{bbold}{U}{bbold}{m}{n}
\DeclareSymbolFontAlphabet{\mathbbold}{bbold}
\newcommand{\1}{\mathbbold{1}}

\newcommand\blfootnote[1]{%
  \begingroup
  \renewcommand\thefootnote{}\footnote{#1}%
  \addtocounter{footnote}{-1}%
  \endgroup
}

\usepackage{color}

\newcommand{\cR}{\mathcal{R}}

\newcommand{\bolda}{\mathbf{a}}
\newcommand{\boldb}{\mathbf{b}}
\newcommand{\boldc}{\mathbf{c}}
\newcommand{\boldd}{\mathbf{d}}

\newcommand{\boldh}{\mathbf{h}}

\newcommand{\boldm}{\mathbf{m}}

\newcommand{\boldv}{\mathbf{v}}

\newcommand{\boldx}{\mathbf{x}}

\makeatletter
\def\namedlabel#1#2{\begingroup
	\def\@currentlabel{#2}%
	\label{#1}\endgroup
}
\makeatother
\begin{document}
\title{Optimal $k$-Deletion Correcting Codes}
\author{\textbf{Jin Sima} \IEEEauthorblockN{and \textbf{Jehoshua Bruck}}\\
	\IEEEauthorblockA{
	Department of Electrical Engineering,
California Institute of Technology, Pasadena 91125, CA, USA\\}}

\maketitle

\begin{abstract}
Levenshtein introduced the problem of constructing $k$-deletion correcting codes  in 1966, proved that the optimal redundancy of those codes is $O(k\log N)$, and proposed an optimal redundancy single-deletion correcting code (using the so-called VT construction). However, the problem of constructing optimal redundancy $k$-deletion correcting codes remained open. Our key contribution is a solution to this longstanding open problem. We present a $k$-deletion correcting code that has redundancy $8k\log n +o(\log n)$ and encoding/decoding algorithms of complexity $O(n^{2k+1})$ for constant~$k$.
\end{abstract}
\blfootnote{This work was presented in part at the IEEE International Symposium on Information Theory, Paris, France, July 2019. }
\section{Introduction}
A set of binary vectors of length $N$ is a $k$-deletion code (denoted by~$\mathcal{C}$) iff any two vectors in ~$\mathcal{C}$ do not share a subsequence of length~$N-k$. The problem of constructing a $k$-deletion code was introduced by Levenshtein~\cite{levenshtein1966binary}. He proved that the optimal redundancy (defined as~$N- \log |\mathcal{C}|$) is~$O(k\log N)$. Specifically, it is in the range~$k\log N +o(\log N)$ to ~$2k\log N +o(\log N)$. In addition, he proposed the following optimal construction (the well-known Varshamov-Tenengolts (VT) code~\cite{vt1965}):
\begin{align}\label{equation:VTconstruction}
\left\{(c_1,\ldots,c_N): \sum^N_{i=1}ic_i\equiv 0 \bmod (N+1)\right\},
\end{align}
that is capable of correcting a single deletion with redundancy not more than~$\log (N+1)$~\cite{levenshtein1966binary}. The encoding/decoding complexity of VT codes is linear in $N$. Generalizing the VT construction to correct more than a single deletion was elusive for more than 50 years. In particular, the past approaches~\cite{helberg2002multiple},~\cite{abdel2012helberg},~\cite{paluncic2012multiple} result in asymptotic code rates that are bounded away from~$1$. 

A recent breakthrough paper~\cite{brakensiek2016efficient} proposed a~$k$-deletion code construction with~$O(k^2\log k \log N)$ redundancy and~$O_k(N\log^4 N)$\footnote{The notion~$O_k$ denotes \textit{parameterized complexity}, i.e., $O_k(N\log^4 N)=f(k)O(N\log^4 N)$ for some function~$f$.} encoding/decoding complexity. For the case~$k=2$ deletions, the redundancy was improved in~\cite{Ryan,sima2018two}. Specifically, the code in~\cite{sima2018two} has redundancy of $7\log N$ and linear encoding/decoding complexity. The work in~\cite{hanna2018guess} considered correction with high probability and proposed a~$k$-deletion code construction with redundancy~$(k+1)(2k+1)\log N + o(\log N)$ and encoding/decoding complexity~$O(N^{k+1}/\log^{k-1}N)$. The result for this randomized coding setting was improved in~\cite{Haeupler}, where redundancy~$O(k\log (n/k))$ and complexity~$poly(n,k)$ were achieved. However, finding a deterministic~$k$-deletion code construction that achieves the optimal order redundancy~$O(k\log N)$ remained elusive. 

Our key contribution is a solution to this longstanding open problem: We present a code construction that achieves~$O(k\log N)$ redundancy and~$O(N^{2k+1})$ encoding/decoding computational complexity (note that the complexity is polynomial in~$N$). 
The following theorem summarizes our main result. We note that in this paper, the optimality of the code is redundancy-wise rather than cardinality-wise, i.e., the result focus on asymptotic rather than exact size of the code. 
The problem of finding optimal cardinality~$k$ deletion code appears highly nontrivial even for~$k=1$.
\begin{theorem}\label{theorem:main}
    For~any integer~$n >k$ and~$N=n+8k\log n +o(\log n)$, there exists an encoding function~$\mathcal{E}:\{0,1\}^n\rightarrow \{0,1\}^{N}$, computed in~$O(n^{2k+1})$ time, and a decoding function~$\mathcal{D}:\{0,1\}^{N-k}\rightarrow \{0,1\}^{n}$, computed in~$O(n^{k+1})$ time, such that for any~$\boldc\in\{0,1\}^n$ and subsequence~$\boldd\in \{0,1\}^{N-k}$ of~$\mathcal{E}(\boldc)$, we have that~$\mathcal{D}(\boldd)=\boldc$. 
\end{theorem}
Recently, an independent work~\cite{kuan2018deterministic} proposed a~$k$ deletion code with~$O(k\log n)$ redundancy and better complexity of~$poly(n,k)$. Compare to the constant~$8k\log n$ in this paper, the constant in~\cite{kuan2018deterministic} is not explicitly given and is at least~$200k\log n$.
Moreover, the approaches in~\cite{kuan2018deterministic} and this paper are different. 

Next we identify and describe our key ideas. The key building blocks in our code construction are: (i) \textit{generalizing the  VT construction} to $k$ deletions by considering constrained sequences, (ii) separating the encoded vector to blocks and using \textit{concatenated codes} and (iii) a novel strategy to \textit{separate the vector to blocks by a single pattern.}

In our previous work for 2-deletions codes~\cite{sima2018two}, we \textit{generalized the VT construction}. In particular, we proved that while the higher order parity checks~$\sum^n_{i=1}i^jc_i  \bmod( n^j+1)$,~$j=0,1,\ldots,t$ might not work in general, those parity checks work in the two deletions case when  the sequences are constrained to have no adjacent~$1$'s. 
In this paper we generalize this idea, specifically, the higher order parity checks work for $k=t/2$ deletions when the sequences we need to protect satisfy the \textit{following constraint: } The distance between any two adjacent ~$1$'s is at least~$k$. 

The fact that we can correct $k$ deletions using the generalization of the VT construction on constrained sequences, enables a \textit{concatenated code construction}, which separates the sequence~$\boldc$ into small blocks. Each block is protected by an inner code, usually a~$k$-deletion code. All the blocks together are protected by an outer code, for example, a Reed-Solomon code. Separating and identifying the boundaries between blocks is one of the main challenges in the concatenated code construction. The work in~\cite{schulman1999asymptotically,guruswami2017deletion} resolved this issue by inserting markers between blocks.
In \cite{brakensiek2016efficient}, an approach that uses occurrences of short subsequences, called patterns, as markers was proposed. The success of decoding in existing approaches requires that the patterns can not be destroyed or generated by~$k$ deletions / insertions.  

Here, we improve the redundancy in~\cite{brakensiek2016efficient} by using \textit{a single pattern to separate the blocks} and allowing it to be destroyed or generated by deletions / insertions.  
The pattern, which we call~
\emph{synchronization pattern}, is a length~$3k+\lceil \log k\rceil +4$ sequence~$\bolda=(a_1,\ldots,a_{3k+\lceil \log k\rceil +4})$ satisfying
\begin{itemize}
    \item $a_{3k+i}=1$ for~$i\in[0,\lceil \log k\rceil +4]$, where~$[0,\lceil \log k\rceil +4]=\{0,\ldots, \lceil \log k\rceil +4\}$.
    \item There does not exist a~$j\in [1,3k-1]$, such that $a_{j+i}=1$ for~$i\in [0, \lceil \log k\rceil +4]$.
\end{itemize}
Namely, 
a \emph{synchronization pattern} is a sequence that ends with~$\lceil \log k\rceil +5$ consecutive~$1$'s and no other 1 run with length~$\lceil \log k\rceil +5$ exists.
For a sequence~$\boldc=(c_1,\ldots,c_n)$, define a \emph{synchronization vector}~$\1_{sync}(\boldc)\in\{0,1\}^n$ by
\begin{align*}
\1_{sync}(\boldc)_i &= \begin{cases}
1, &\text{if~$(c_{i-3k+1},c_{i-3k+2},\ldots,c_{i+\lceil \log k\rceil +4})$ is a \emph{synchronization pattern},}\\
0,&\text{else.}\\
\end{cases}
\end{align*}
Note that~$\1_{sync}(\boldc)_i=0$ for~$i\in [1,3k-1]$ and for~$i\in [n-\lceil \log k\rceil -3,n]$. 
It can be seen from the definition that any two consecutive~$1$ entries in~$\1_{sync}(\boldc)$ have distance at least~$3k$.

Now we are ready to describe our construction that is a generalization of the VT code. Define the integer vectors
\begin{align*}
    \boldm^{(\ell)}\triangleq (1^\ell,1^\ell+2^\ell,\ldots,\sum^n_{j=1}j^\ell)
\end{align*}
for~$\ell\in [0,\ldots,6k]$, where the~$i$-th entry of~$\boldm^{(\ell)}$ is the sum of the~$\ell$-th powers of the first~$i$ entries. 
Given a sequence~$\boldc\in \{0,1\}^n$ we compute a (VT like)  redundancy of dimension~$6k+1$ as follows:
\begin{align}\label{equation:deff}
    &f(\boldc)_{\ell}\triangleq \boldc \cdot \boldm^{(\ell)} \bmod 3kn^{\ell+1},
\end{align}
for~$\ell\in [0,6k]$.
It will be shown that
the vector~$f(\1_{sync}(\boldc))$ helps protect the \emph{synchronization vector}~$\1_{sync}(\boldc)$ from~$k$ deletions in~$\boldc$. 

The rest of the paper is organized as follows. Section~\ref{section:outline} provides an outline of our construction and some of the basic lemmas. Section~\ref{section:synchronizationvectors} presents our VT generalization for recovering the \emph{synchronization vector}. Section~\ref{section:kdensehash} explains how to correct~$k$ deletions based on the \emph{synchronization vector}, when the \emph{synchronization patterns} appear frequently. Section~\ref{section:transformation} describes an algorithm to transform a sequence into one with dense \emph{synchronization patterns}. Section~\ref{section:encoding} presents the encoding and decoding of the code. Section~\ref{section:conclusion} concludes the paper. 

\section{Outline and Preliminaries}\label{section:outline}
In this section we give an overview of the ingredients that constitute our code construction and the existing results that will be used. For a sequence~$\boldc\in\{0,1\}^n$, define its deletion ball $B_k(\boldc)$ as the collection of sequences that share a length~$n-k$ subsequence with~$\boldc$. 
We first present a lemma showing that
the~\emph{synchronization vector}~$\1_{sync}(\boldc)$ can be recovered from~$k$ deletions with the help of 
~$f(\1_{sync}(\boldc))$. Its proof will be given in Section~\ref{subsection:proofoflemma1}.
\begin{lemma}\label{lemma:syncvector}
For integers~$n$ and~$k$ and sequences~$\boldc,\boldc'$, if $\boldc'\in B_k(\boldc)$~and~$f(\1_{sync}(\boldc))=f(\1_{sync}(\boldc'))$, then~$\1_{sync}(\boldc)=\1_{sync}(\boldc')$.
\end{lemma}
Lemma~\ref{lemma:syncvector} implies that~$\1_{sync}(\boldc)$ can be protected using~$O(k^2\log n)$ bit redundancy~$f(\1_{sync}(\boldc))$. To further reduce the redundancy and get it down to~$O(k\log n)$, we apply modulo operations on~$f(\1_{sync}(\boldc))$.
For an integer vector~$\boldv=(v_0,\ldots,v_{6k})$ that satisfies~$0\le v_e< 3kn^{e+1}$,~$e\in [0,6k]$, 
let
\begin{align}\label{equation:Mfunction}
    M(\boldv)= \sum^{6k}_{e=0}v_e\prod^{e-1}_{i=0}3kn^{i+1}
\end{align}
be a one-to-one mapping that maps the vector~$\boldv$ into an integer~$M(\boldv)\in[0,(3k)^{6k+1}n^{(3k+1)(6k+1)}-1]$.
Then we have the following lemma, which will be proved in Section~\ref{subsection:proofoflemma2}. 
\begin{lemma}\label{lemma:protectingpattern}
For integers~$n$ and~$k$, there exists a function~$p:\{0,1\}^{n}\rightarrow [1,2^{2k\log n +o(\log n)}]$, such that
if~$M(f(\1_{sync}(\boldc)))\equiv M(f(\1_{sync}(\boldc')))\bmod p(\boldc)$ for two sequences~$\boldc\in\{0,1\}^n$ and~$\boldc'\in B_k(\boldc)$, then~$\1_{sync}(\boldc)=\1_{sync}(\boldc')$.
Hence if
\begin{align*}
(M(f(\1_{sync}(\boldc))) \bmod p(\boldc),p(\boldc))   
=(M(f(\1_{sync}(\boldc'))) \bmod p(\boldc'),p(\boldc'))
\end{align*}
and~$\boldc'\in B_k(\boldc)$, we have that~$\1_{sync}(\boldc)=\1_{sync}(\boldc')$.
\end{lemma}
Lemma~\ref{lemma:protectingpattern} presents a~$4k\log n +o(\log n)$ bit hash for correcting~$\1_{sync}(\boldc)$. With the knowledge of the~\emph{synchronization vector}~$\1_{sync}(\boldc)$,
the next lemma shows that the sequence~$\boldc$ can be further recovered using another~$4k\log +o(\log n)$ bit hash, when~$\boldc$ satisfies a "\emph{$k$ dense}" property. The proof of Lemma~\ref{lemma:recoverforkdense} will be given in Section~\ref{section:kdensehash}.

A sequence~$\boldc\in\{0,1\}^n$ is said to be~\emph{$k$ dense} if the lengths of the~$0$ runs in~$\1_{sync}(\boldc)$ is at most~
\begin{align*}
L\triangleq (\lceil \log k\rceil+5)2^{\lceil \log k\rceil+9}\lceil \log n \rceil
+ (3k+\lceil \log k\rceil+4)(\lceil \log n\rceil +9+\lceil\log k\rceil).
\end{align*}
For~$k$ \emph{dense}~$\boldc$, the distance between two consecutive~$1$ entries in~$\1_{sync}(\boldc)$ is at most~$L+1$.
\begin{lemma}\label{lemma:recoverforkdense}
For integers~$k$ and~$n>k$, there exists a function~$Hash_k:\{0,1\}^n\rightarrow \{0,1\}^{4k\log n+o(\log n)}$, such that every~\emph{$k$ dense} sequence~$\boldc\in\{0,1\}^n$ can be recovered, given its \emph{synchronization vector}~$\1_{sync}(\boldc)$,  its length~$n-k$ subsequence~$\boldd$,  and~$Hash_k(\boldc)$.
\end{lemma}
Finally, to encode for arbitrary sequence~$\boldc\in\{0,1\}^n$, a mapping that transforms any sequence to a~\emph{$k$ dense} sequence is given in the following lemma. The details will be given in Section~\ref{subsection:Tfunction}.
\begin{lemma}\label{lemma:transformation}
For integers~$k$ and~$n>k$, there exists a map
$T:\{0,1\}^n\rightarrow\{0,1\}^{n+3k+3\lceil\log k\rceil+15}$, computable in~$poly(n,k)$ time, such that~$T(\boldc)$ is a~\emph{$k$ dense} sequence for~$\boldc\in\{0,1\}^n$. Moreover, the sequence~$\boldc$ can be recovered from~$T(\boldc)$. 
\end{lemma}
The next lemmas are from existing results. 
Lemma~\ref{lemma:hash} gives a~$k$ deletion correcting hash function for short sequences.
It is an extension of the result in~\cite{brakensiek2016efficient}. 
Lemma~\ref{lemma:bkball} is a slight variation of the result in~\cite{levenshtein1966binary}.  It shows the equivalence between correcting deletions and correcting deletions and insertions.  Lemma~\ref{lemma:numberofdivisors} (See \cite{Nicolas}) gives an upper bound on the number of divisors of a positive integer~$n$. Lemma~\ref{lemma:bkball} and Lemma~\ref{lemma:numberofdivisors} will be used in proving Lemma~\ref{lemma:protectingpattern}.
\begin{lemma}\label{lemma:hash}
Let~$k$ be a fixed integer. For integers~$w$ and~$n$, there exists a hash function~$H:\{0,1\}^w\rightarrow$\newline$ \{0,1\}^{\lceil (w/\lceil \log n \rceil) \rceil (2k\log \log n +O(1))}$, computable in~$O_k((w/ \log n)  n\log^{2k} n )$ time, such that any sequence~$\boldc\in\{0,1\}^w$ can be recovered from its length~$w-k$ subsequence~$\boldd$ and the hash~$H(\boldc)$.
\end{lemma}
\begin{proof}
We first show that there exists a hash function~$H': \{0,1\}^{\lceil \log n \rceil}\rightarrow\{0,1\}^{2k\log\log n +O(1)}$, computable in~$O_k(n \log^{2k}n)$ time, that protects a sequence~$\boldc'\in \{0,1\}^{\lceil \log n\rceil}$ from~$k$ deletions. Note that~$|B_k(\boldc')|\le \binom{\lceil \log n\rceil}{k}^22^k\le 2\lceil\log n\rceil^{2k}$. Hence it suffices to use brute force and assign hash values for each possible~$\boldc'$, and compare with all sequences in~$B_k(\boldc')$. The total complexity is~$O_k(n\log^{2k}n)$ and the size of the hash value~$H(\boldc')$ is~$\log (|B_k(\boldc')|+1)\le 2k\log\log n+O(1)$. 

Now split~$\boldc$ into~$\lceil(w/\lceil \log n\rceil)\rceil$ blocks~$c_{(i-1)\lceil \log n \rceil+1},\ldots,c_{i\lceil \log n \rceil}$,~$i\in [1,\lceil(w/\lceil \log n\rceil)\rceil]$ of length~$\lceil \log n \rceil$. If the length of the last block is less than~$\lceil \log n \rceil$, add zeros to the end of the last block such that its length is~$\lceil \log n \rceil$.
Assign a hash value~$\boldh_i=H'((c_{(i-1)\lceil \log n \rceil+1},\ldots,c_{i\lceil \log n \rceil}))$,~$i\in [1,\lceil(w/\lceil \log n\rceil)\rceil]$ for each block. Let~$H(\boldc)=(\boldh_1,\ldots,\boldh_{\lceil(w/\lceil \log n\rceil)\rceil})$ be the concatenation of~$\boldh_i$ for~$i\in [1,\lceil(w/\lceil \log n\rceil)\rceil]$. The length of~$H(\boldc)$ is~$\lceil (w/\lceil \log n \rceil) \rceil (2k\log \log n +O(1))$ and the complexity of~$H(\boldc)$ is~$O_k((w/ \log n)  n\log^{2k} n )$

Let~$\boldd$ be a length~$n-k$ subsequence of~$\boldc$. Then~$d_{(i-1)\lceil \log n \rceil+1},\ldots,d_{i\lceil \log n \rceil-k}$ is a length~$\lceil \log n \rceil-k$ subsequence of the block~$c_{(i-1)\lceil \log n \rceil+1},\ldots,c_{i\lceil \log n \rceil}$. Hence the~$i$-th block~$c_{(i-1)\lceil \log n \rceil+1},\ldots,c_{i\lceil \log n \rceil}$ can be recovered from~$\boldh_i$ and \newline $d_{(i-1)\lceil \log n \rceil+1},\ldots,d_{i\lceil \log n \rceil-k}$. Therefore,~$\boldc$ can be recovered given~$\boldd$ and~$H(\boldc)$ and the proof is done.
\end{proof}
\begin{lemma}\label{lemma:bkball}
Let~$r$,~$s$, and~$k$ be integers satisfying~$r+s\le k$. For sequences~$\boldc,\boldc'\in\{0,1\}^n$, if~$\boldc'$ and~$\boldc$ share a common resulting sequence after~$r$ deletions and~$s$ insertions in~both, then~$\boldc'\in B_k(\boldc)$.
\end{lemma}
\begin{lemma}\label{lemma:numberofdivisors}
For a positive integer~$n\ge 3$, the number of divisors of~$n$ is upper bounded by~$2^{1.6\ln n/(\ln \ln n)}$.
\end{lemma}
The proofs of Lemma~\ref{lemma:syncvector},~Lemma~\ref{lemma:protectingpattern},~Lemma~\ref{lemma:recoverforkdense}, and Lemma~\ref{lemma:transformation} rely on several propositions that will be presented in the next sections. For convenience, a dependency graph for the theorem, lemmas and propositions is given in Fig.~\ref{figure:dependencies}.  
\tikzstyle{theorem} = [diamond, draw, fill=green!20, text width=4.5em, text badly centered, node distance=2.5cm, inner sep=0pt]
\tikzstyle{lemma} = [rectangle, draw, fill=blue!20, text width=5em, text centered, rounded corners, minimum height=2em]
\tikzstyle{line} = [draw, -latex']
\tikzstyle{proposition} = [draw, ellipse,fill=red!20, node distance=3cm, minimum height=2em]
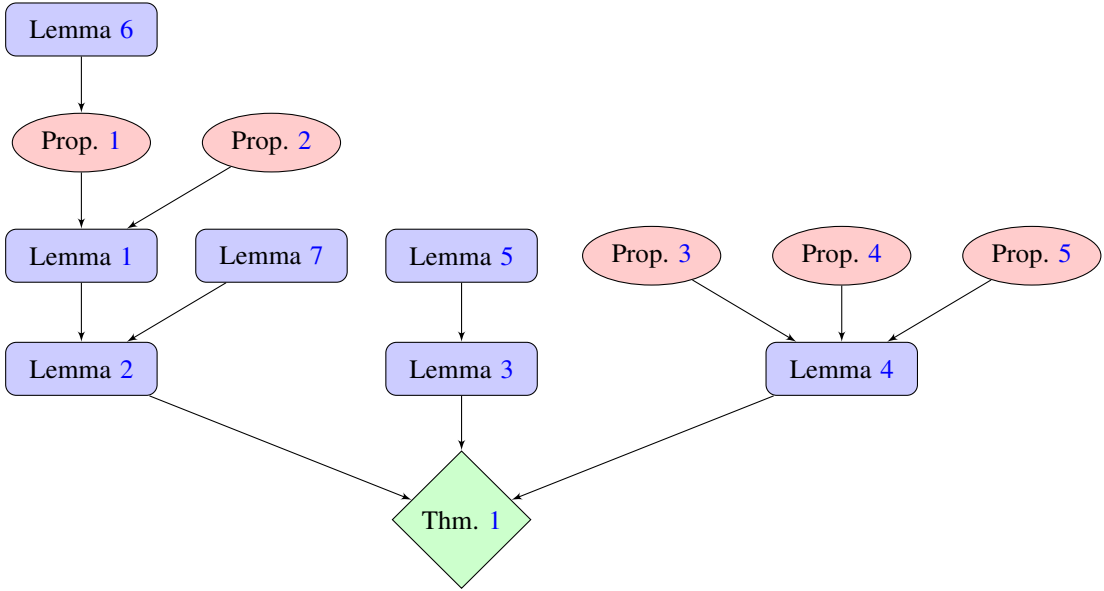
\begin{figure}
    \centering
    \begin{tikzpicture}[node distance = 1.5cm, auto]
    \node[lemma](bkball){Lemma~\ref{lemma:bkball}};
    \node[proposition, below of = bkball, node distance = 1.5cm](editdistance){Prop.~\ref{lemma:editdistance}};
    \node[proposition, right of = editdistance, node distance =2.5cm](vtextension){Prop.~\ref{lemma:vtextension}};
    \node[lemma, below of=editdistance](syncvector){Lemma~\ref{lemma:syncvector}};
    \node[lemma, right of=syncvector,node distance =2.5cm](numberofdivisors){Lemma~\ref{lemma:numberofdivisors}};
    \node[lemma, right of=numberofdivisors,node distance =2.5cm](hash){Lemma~\ref{lemma:hash}};
    \node[proposition, right of=hash,node distance =2.5cm](property1and2){Prop.~\ref{lemma:property1and2}};
    \node[proposition, right of=property1and2,node distance =2.5cm](frequentpattern){Prop.~\ref{lemma:frequentpattern}};
    \node[proposition, right of=frequentpattern,node distance =2.5cm](3kto3kminus1){Prop.~\ref{lemma:3kto3kminus1}};
    \node[lemma, below of=syncvector, node distance = 1.5cm](protectingpattern){Lemma~\ref{lemma:protectingpattern}};
	\node[lemma, below of=hash, node distance = 1.5cm](recoverforkdense){Lemma~\ref{lemma:recoverforkdense}};
    \node[lemma,  below of=frequentpattern, node distance = 1.5cm](transformation){Lemma~\ref{lemma:transformation}};
    \node[theorem, below of=recoverforkdense, node distance = 2cm](main){Thm.~\ref{theorem:main}};
   
    \path [line] (bkball) -- (editdistance);
    \path [line] (editdistance) -- (syncvector);
    \path [line] (vtextension) -- (syncvector);
    \path [line] (syncvector) -- (protectingpattern);
    \path [line] (numberofdivisors) -- (protectingpattern);    
    \path [line] (hash) -- (recoverforkdense);
    \path [line] (property1and2) -- (transformation);
    \path [line] (frequentpattern) -- (transformation);
    \path [line] (3kto3kminus1) -- (transformation);
    \path [line] (protectingpattern) -- (main);
    \path [line] (recoverforkdense) -- (main);
    \path [line] (transformation) -- (main);
\end{tikzpicture}
    \caption{Dependencies of the claims in the paper.}
    \label{figure:dependencies}
\end{figure}

\section{Protecting the synchronization vectors}\label{section:synchronizationvectors}
In this section we present a hash function to protect the synchronization vector~$\1_{sync}(\boldc)$ from~$k$ deletions and prove Lemma~\ref{lemma:protectingpattern}.
We first prove Lemma~\ref{lemma:syncvector},  based on Proposition~\ref{lemma:editdistance} and Proposition~\ref{lemma:vtextension}. In Proposition~\ref{lemma:editdistance} we present a bound on the radius of the~$B_k$ ball for the synchronization vector. In Proposition~\ref{lemma:vtextension}, we prove that the higher order parity check helps correct multiple deletions for sequences in which the distance between two consecutive~$1$'s is at least~$3k$. Since~$\1_{sync}(\boldc)$ is such a sequence, we conclude that the higher order parity check helps recover~$\1_{sync}(\boldc)$. After proving Lemma~\ref{lemma:syncvector}, we finally use Lemma~\ref{lemma:numberofdivisors} to further compress the size of the hash function that protects~$\1_{sync}(\boldc)$ and thus prove Lemma~\ref{lemma:protectingpattern}.
\begin{proposition}\label{lemma:editdistance}
For~$\boldc,\boldc'\in\{0,1\}^n$, if~$\boldc'\in B_k(\boldc)$, then~$\1_{sync}(\boldc')\in B_{3k}(\1_{sync}(\boldc))$.
\end{proposition}
\begin{proof}
Since~$\boldc'\in B_k(\boldc)$, the sequences~$\boldc'$ and~$\boldc$ share a common subsequence after~$k$ deletions in both. We now show that a single deletion in~$\boldc$ causes at most two deletions and one insertion in its \emph{synchronization vector}~$\1_{sync}(\boldc)$. 
We first show that a deletion in~$\boldc$ can destroy or generate at most~$1$ \emph{synchronization pattern}. This is because the deletion occurs within the \emph{synchronization pattern} destroyed or generated. Otherwise the \emph{synchronization pattern} is not destroyed or generated by the deletion.
Therefore, we need to consider four cases in total. Let~$\boldd'$ be the subsequence of~$\boldc$ after a single deletion.
\begin{enumerate}
\item The deletion destroys a \emph{synchronization pattern}~$(c_{i+1},\ldots,c_{i+3k+\lceil \log k\rceil +4})$ for some~$i$ and no \emph{synchronization pattern} is generated. Then the sequence~$\1_{sync}(\boldd')$ can be obtained by deleting the~$1$ entry~$\1_{sync}(\boldc)_{i+3k}$ in~$\1_{sync}(\boldc)$.
\item The deletion generates a new \emph{synchronization pattern}~$(c'_{i'+1},\ldots,c'_{i'+3k+\lceil \log k\rceil +4})$ for some~$i'$ and destroys a \emph{synchronization pattern}~$(c_{i+1},\ldots,c_{i+3k+\lceil \log k\rceil +4})$. The sequence~$\1_{sync}(\boldd')$ can be obtained by deleting the~$1$ entry~$\1_{sync}(\boldc)_{i+3k}$ and the~$0$ entry~$\1_{sync}(\boldc)_{i+3k-1}$ in~$\1_{sync}(\boldc)$ and inserting a~$1$ entry at~$\1_{sync}(\boldc)_{i'+3k}$.
\item The deletion generates a new \emph{synchronization pattern}~$(c'_{i'+1},\ldots,c'_{i'+3k+\lceil \log k\rceil +4})$ for some~$i'$ and no \emph{synchronization pattern} is destroyed. Then the
~$\1_{sync}(\boldd')$ can be obtained by deleting  two~$0$ entries~$\1_{sync}(\boldc)_{i'+3k}$ and~$\1_{sync}(\boldc)_{i'+3k+1}$ in~$\1_{sync}(\boldc)$ and inserting a~$1$ entry at~$\1_{sync}(\boldc)_{i'+3k}$.
\item No \emph{synchronization pattern} is generated or destroyed. Then~$\1_{sync}(\boldd')$ can be obtained by deleting a~$0$ entry~$\1_{sync}(\boldc)_{j}$, where~$j$ is the location of the deletion.  
\end{enumerate}
In a summary, in each of the above cases, a single deletion in~$\boldc$ causes at most two deletions and one insertion in~$\1_{sync}(\boldc)$. 
Hence~$k$ deletions in~$\boldc$ and~$\boldc'$ cause at most~$2k$ deletions and~$k$ insertions in~$\1_{sync}(\boldc)$ and~$\1_{sync}(\boldc')$ respectively. According to Lemma~\ref{lemma:bkball}, we have that~$\1_{sync}(\boldc')\in B_{3k}(\1_{sync}(\boldc))$ when~$\boldc'\in B_k(\boldc)$. Hence Proposition~\ref{lemma:editdistance} is proved.
\end{proof}
Let~$\mathcal{R}_m$ be the set of length~$n$ sequences in which the~$0$ runs have length at least~$m-1$, meaning that any two consecutive~$1$ entries in a sequence $\boldc\in\mathcal{R}_m$ have distance at least~$m$. 
\begin{proposition}\label{lemma:vtextension}
For~$\boldc,\boldc'\in \mathcal{R}_{3k}$, if~$\boldc'\in B_{3k}(\boldc)$ and~$ \boldc \cdot \boldm^{(e)}= \boldc' \cdot \boldm^{(e)}$ (recall that~$f(\boldc)_e=\boldc \cdot \boldm^{(e)}\bmod 3kn^{e+1}$) for~$e\in [0,,6k]$, then~$\boldc=\boldc'$.
\end{proposition}
\begin{proof}
We first compute the difference~$ \boldc \cdot \boldm^{(e)}- \boldc' \cdot  \boldm^{(e)}$.
Since~$\boldc'\in B_{3k}(\boldc)$, there exist two subsets~$\boldsymbol{\delta}=\{\delta_{1},\ldots,\delta_{3k}\}\subset \{1,\ldots,n\}$ and~$\boldsymbol{\delta}'=\{\delta'_{1},\ldots,\delta'_{3k}\}\subset \{1,\ldots,n\}$ such that deleting bits with indices~$\boldsymbol{\delta}$ and~$\boldsymbol{\delta}'$ respectively from~$\boldc$ and~$\boldc'$ result in the same length~$n-3k$ subsequence, i.e.,~$(c_i:i\notin\boldsymbol{\delta})=(c'_i:i\notin\boldsymbol{\delta'})$. Let~$\boldsymbol{\Delta}=\{i:c_i=1\}$ and~$\boldsymbol{\Delta}'=\{i:c'_i=1\}$ be the indices of~$1$ entries in~$\boldc$ and~$\boldc'$ respectively.
Let~$S_1=\boldsymbol{\Delta}\cap \boldsymbol{\delta}$ 
be the indices of~$1$ entries that are deleted in~$\boldc$. Then
$S^c_1=\boldsymbol{\Delta}\cap ([1,n]\backslash \boldsymbol{\delta})$ denotes the indices of~$1$ entries  that are not deleted. Similarly let~$S_2=\boldsymbol{\Delta}'\cap \boldsymbol{\delta}'$ and~$S^c_2=\boldsymbol{\Delta}'\cap ([1,n]\backslash \boldsymbol{\delta}')$ be the indices of $1$ entries that are deleted and not in~$\boldc'$ respectively.
Let the elements in~$\boldsymbol{\delta}\cup \boldsymbol{\delta'}$ be ordered by~$1\le p_1\le p_2\le\ldots\le p_{6k}\le n$. Denoting~$p_0=0$ and~$p_{6k+1}=n$, we have that
\begin{align}\label{equation:difference}
    \boldc \cdot \boldm^{(e)}- \boldc' \cdot  \boldm^{(e)}&=\sum_{\ell\in \boldsymbol{\Delta}}\boldm^{(e)}_\ell - \sum_{\ell\in\boldsymbol{ \Delta}'}\boldm^{(e)}_\ell
    \nonumber\\
    &=\sum_{\ell\in \boldsymbol{\Delta}}(\sum^\ell_{i=1}i^e) - \sum_{\ell\in\boldsymbol{ \Delta}'}(\sum^\ell_{i=1}i^e)\nonumber\\
    &=\sum^n_{i=1}(\sum_{\ell\in \boldsymbol{\Delta}\cap[i,n]}i^e) - \sum^n_{i=1}(\sum_{\ell\in \boldsymbol{\Delta'}\cap [i,n]}i^e)\nonumber\\
    &=\sum^n_{i=1}(|\boldsymbol{\Delta}\cap[i,n]|-|\boldsymbol{\Delta}'\cap[i,n]|)i^e \nonumber\\
    &= \sum^n_{i=1}(|S_1\cap[i,n]|+|S^c_1\cap[i,n]|-|S_2\cap[i,n]|-|S^c_2\cap[i,n]|)i^e\nonumber\\
    &=\sum^{6k}_{j=0}\sum^{p_{j+1}}_{i=p_{j}+1}(|S_1\cap[i,n]|-|S_2\cap[i,n]|+|S^c_1\cap[i,n]|-|S^c_2\cap[i,n]|)i^e\nonumber\\
    &\overset{(1)}{=}\sum^{6k}_{j=0}\sum^{p_{j+1}}_{i=p_{j}+1}(|S_1\cap[p_{j+1},n]|-|S_2\cap[p_{j+1},n]|+|S^c_1\cap[i,n]|-|S^c_2\cap[i,n]|)i^e,
\end{align}
where~$(1)$ holds since there is no~$1$ entry in interval~$(p_j,p_{j+1})$.
In the following we show
\begin{enumerate}
    \item $-1\le |S^c_1\cap[i,n]|-|S^c_2\cap[i,n]|\le 1$ for~$i\in [1,n]$.
    \item  For each interval~$(p_j,p_{j+1}]$,~$j=0,\ldots,6k$, we have either~$|S^c_1\cap[i,n]|-|S^c_2\cap[i,n]|\le 0$ for all~$i\in (p_j,p_{j+1}]$ or
$|S^c_1\cap[i,n]|-|S^c_2\cap[i,n]|\ge 0$
for all~$i\in (p_j,p_{j+1}]$. 
\end{enumerate}
\emph{Proof of (1)}: 
Since deleting bits with indices~$\boldsymbol{\delta}$ in~$\boldc$ and deleting bits with indices~$\boldsymbol{\delta'}$ in~$\boldc'$ result in the same subsequence.
Hence for every~$i\in S^c_1$, there is a unique corresponding index~$i'\in S^c_2$ such that the two~$1$ entries~$c_i$ and~$c'_{i'}$ end in the same location after deletions, i.e.,~$i-|\boldsymbol{\delta}\cap [1,i-1]|=i'-|\boldsymbol{\delta}'\cap [1,i'-1]|$. This implies that~$|i'-i|\le 3k$. 
Fix integers~$i$ and~$i'$. Then by definition of~$i$ and~$i'$, for every~$x\in S^c_1\cap [i+1,n]$ there is a unique corresponding~$y\in S^c_2\cap[i'+,n]$ such that the two~$1$ entries~$c_x$ and~$c'_y$ end in the same location after deletions. Hence we have that~$|S^c_1\cap[i+1,n]|=|S^c_2\cap[i'+1,n]|$ and that~$|S^c_1\cap[i,n]|=|S^c_2\cap[i',n]|$. Without loss of generality assume that~$i'\ge i$. Then,
\begin{align*}
    |S^c_1\cap[i,n]|-|S^c_2\cap[i,n]|&=|S^c_1\cap[i,n]|-|S^c_2\cap[i',n]|-
    |S^c_2\cap[i+1,i']|\\
    &=-|S^c_2\cap[i+1,i']|\\
    &\overset{(a)}{\ge} -|S^c_2\cap[i+1,i+3k]|\\
    &\overset{(b)}{\ge} -1,
\end{align*}
where~$(a)$ follows from the fact that~$|i'-i|\le 3k$ and~$(b)$ follows from the fact that~$\boldc,\boldc'\in\cR_{3k}$. Similarly, when~$i'\le i$, we have that~$|S^c_1\cap[i,n]|-|S^c_2\cap[i,n]|\le 1$. This completes the proof of~$(1)$

We now prove~$(2)$ by contradiction. Supposed on the contrary, there exist~$i_1,i_2\in (p_j,p_{j+1}]$ such that~$i_1<i_2$ and
\begin{align*}
    (|S^c_1\cap[i_1,n]|-|S^c_2\cap[i_1,n]|)(|S^c_1\cap[i_2,n]|-|S^c_2\cap[i_2,n]|)<0
\end{align*}
By symmetry it can be assumed that~$|S^c_1\cap[i_1,n]|-|S^c_2\cap[i_1,n]=-1$ and~$|S^c_1\cap[i_2,n]|-|S^c_2\cap[i_2,n]|=1$. Note that~$|S^c_1|=|S^c_2|$ and there exists a one-to-one correspondence between the elements in~$S^c_1$ and the elements in~$S^c_2$. Hence from~$|S^c_1\cap[i_1,n]|-|S^c_2\cap[i_1,n]=-1$, there exist two integers~$x\in S^c_1\cap[i_1-1,n]$
and~$y\in S^c_2\cap[i_1,n]$ such that the two~$1$ entries~$c_x$ and~$c'_y$ are in the same location after deletions, i.e.,~$x-|\boldsymbol{\delta}\cap[1,x-1]|=y-|\boldsymbol{\delta}'\cap[1,y-1]|$. Therefore, we have that
\begin{align*}
    i_1-|\boldsymbol{\delta}\cap[1,i_1-1]|>&i_1-1-|\boldsymbol{\delta}\cap[1,i_1-1]|\\
    \ge&x-|\boldsymbol{\delta}\cap[1,x-1]|\\
    =&y-|\boldsymbol{\delta}'\cap[1,y-1]|\\
    \ge& i_1-|\boldsymbol{\delta}'\cap[1,i_1-1]|,
\end{align*}
which implies that
\begin{align}\label{equation:contradiction1}
    |\boldsymbol{\delta}\cap[1,i_1-1]|<|\boldsymbol{\delta}'\cap[1,i_1-1]|.
\end{align}
 Similarly, from~$|S^c_1\cap[i_2,n]|-|S^c_2\cap[i_2,n]|=1$ we have that
\begin{align}\label{equation:contradiction2}
|\boldsymbol{\delta}\cap[1,i_2-1]|>|\boldsymbol{\delta}'\cap[1,i_2-1]| .   
\end{align}
Eq.~\eqref{equation:contradiction1} and Eq.~\eqref{equation:contradiction2} implies that
\begin{align}\label{equation:contradiction}
|\boldsymbol{\delta}\cap[1,i_2-1]|-|\boldsymbol{\delta}\cap[1,i_1-1]|\ge &
|\boldsymbol{\delta}'\cap[1,i_2-1]|+1-|\boldsymbol{\delta}'\cap[1,i_1-1]|+1\nonumber\\
\ge & 2.
\end{align}
However, since~$i_1,i_2\in (p_j,p_{j+1}]$, we have that~$|\boldsymbol{\delta}\cap[1,i_1]|=|\boldsymbol{\delta}\cap[1,i_2-1]|$ and~$|\boldsymbol{\delta}'\cap[1,i_1]|=|\boldsymbol{\delta}'\cap[1,i_2-1]|$,
which implies that
\begin{align*}
   |\boldsymbol{\delta}\cap[1,i_2-1]|-|\boldsymbol{\delta}\cap[1,i_1-1]|\le
   &|\boldsymbol{\delta}\cap[1,i_2-1]|-|\boldsymbol{\delta}\cap[1,i_1]|+1\\
   =&1.
\end{align*}
contradicting to Eq.~\eqref{equation:contradiction}. Hence there do not exist different integers~$i_1,i_2\in (p_j,p_{j+1}]$ such that
\begin{align*}
    (|S^c_1\cap[i_1,n]|-|S^c_2\cap[i_1,n]|)(|S^c_1\cap[i_2,n]|-|S^c_2\cap[i_2,n]|)<0.
\end{align*}
Hence~$(2)$ is proved.

Denote
\begin{align*}
s_i\triangleq |S_1\cap[i,n]|-|S_2\cap[i,n]|+|S^c_1\cap[i,n]|-|S^c_2\cap[i,n]|
\end{align*}
Note that~$|S_1\cap[i,n]|-|S_2\cap[i,n]|=|S_1\cap[p_{j+1},n]|-|S_2\cap[p_{j+1},n]|$ for~$i\in (p_j,p_{j+1}]$.
Hence from~$(1)$ and~$(2)$ it follows that for each interval~$(p_j,p_{j+1}]$,~$j\in\{0,\ldots,6k\}$, either~$s_i\ge 0$ for all~$i\in (p_j,p_{j+1}]$ or~$s_i\le 0$ for all~$i\in (p_j,p_{j+1}]$. 
Let~$\boldx=(x_0,\ldots,x_{6k})\in \{-1,1\}^{6k+1}$ be a vector defined by
\begin{align*}
x_i &= \begin{cases}
-1, &\text{if~$s_j<0$ for some~$j\in (p_i,p_{i+1}]$}\\
1,&\text{else.}\\
\end{cases}.
\end{align*}
Then from Eq.~\eqref{equation:difference} we have that
\begin{align}\label{equation:linearequation}
    \boldc \cdot \boldm^{(e)}- \boldc' \cdot  \boldm^{(e)}= \sum^{6k}_{j=0}(\sum^{p_{j+1}}_{i=p_{j}+1}|s_i|i^e)x_j
\end{align}
Let~$A$ be a~$6k+1\times 6k+1$ matrix with entries defined by~$A_{e,j}=\sum^{p_{j}}_{i=p_{j-1}+1}|s_i|i^{e-1}$ for~$e,j\in\{1,\ldots,6k+1\}$.
If~$\boldc \cdot \boldm^{(e)}= \boldc' \cdot  \boldm^{(e)}$ for~$e\in [0,6k]$, we have the following linear equation
\begin{align}\label{equation:linear}
A\boldsymbol{x}=\begin{bmatrix}
\sum^{p_{1}}_{i=p_{0}+1}|s_i|i^0 &\ldots & \sum^{p_{6k+1}}_{i=p_{6k}+1}|s_i|i^0\\
\vdots &\ddots & \vdots\\
\sum^{p_{1}}_{i=p_{0}+1}|s_i|i^{6k} &\ldots & \sum^{p_{6k+1}}_{i=p_{6k}+1}|s_i|i^{6k}
\end{bmatrix}
\begin{bmatrix}
x_0\\
\vdots\\
x_{6k}\\
\end{bmatrix}
=0
\end{align}
We show that this is impossible unless~$A$ is a zero matrix. Suppose on the contrary that~$A$ is nonzero, let~$j_1<\ldots <j_Q$ be the indices of all nonzero columns of~$A$. Let~$B$ be a submatrix of~$A$, obtained by choosing the intersection of the first~$Q$ rows and columns with indices~$j_1,\ldots,j_Q$. Then we have that
\begin{align}\label{equation:linear}
B\boldsymbol{x}'=\begin{bmatrix}
\sum^{p_{j_1}}_{i=p_{j_1-1}+1}|s_i|i^0 &\ldots & \sum^{p_{j_Q}}_{i=p_{j_{Q-1}}+1}|s_i|i^0\\
\vdots &\ddots & \vdots\\
\sum^{p_{j_1}}_{i=p_{j_1-1}+1}|s_i|i^{Q-1} &\ldots & \sum^{p_{j_Q}}_{i=p_{j_{Q-1}}+1}|s_i|i^{Q-1}
\end{bmatrix}
\begin{bmatrix}
x_0\\
\vdots\\
x_{Q-1}\\
\end{bmatrix}
=0
\end{align}
By the multi-linearity of the determinant, 
\begin{align}
\det(B)=&\sum_{i_1\in (p_{j_1-1},p_{j_1}],\ldots,i_M\in (p_{j_Q-1},p_{j_Q}]}\det \begin{pmatrix} |s_{i_1}|i^0_1&\ldots&|s_{i_Q}|i^0_Q\\ \vdots&\ddots&\vdots\\ |s_{i_1}|i_1^{Q-1}&|s_{i_Q}|\ldots&i_Q^{Q-1} \end{pmatrix}\nonumber\\
=&\sum_{i_1\in (p_{j_1-1},p_{j_1}],\ldots,i_Q\in (p_{j_{Q-1}},p_{j_Q}]}[\prod^Q_{q=1} |s_{i_q}| \prod_{1\le m<\ell \le Q}(i_{\ell}-i_{m})]
\end{align}
is positive since~$i_{\ell}>i_{m}$ for~$m>\ell$
and for~$i_1\in (p_{j_1-1},p_{j_1}],\ldots,i_Q\in (p_{j_Q-1},p_{j_Q}]$. Note that all the columns of~$B$ are nonzero. 
Therefore, the linear equation~$B\boldsymbol{x}'=0$ does not have nonzero solutions, contradicting to the fact that~$\boldsymbol{x}'\in\{-1,1\}^Q$.
Hence~$A$ is a zero matrix, meaning that
\begin{align*}
&|S_1\cap[i,n]|-|S_2\cap[i,n]|+|S^c_1\cap[i,n]|-|S^c_2\cap[i,n]|\\
=&|\boldsymbol{\Delta}\cap [i,n]|-|\boldsymbol{\Delta}'\cap [i,n]|=0
\end{align*}
for~$i\in\{1,\ldots,n\}$. This implies~$\boldsymbol{\Delta}=\boldsymbol{\Delta}'$ and thus~$\boldc=\boldc'$. Hence Proposition~\ref{lemma:vtextension} is proved.
\end{proof}
\subsection{Proof of Lemma~\ref{lemma:syncvector}}\label{subsection:proofoflemma1}
From Proposition~\ref{lemma:editdistance} we have that~$\1_{sync}(\boldc')\in B_{3k}(\1_{sync}(\boldc))$. Hence~$(\1_{sync}(\boldc')_i,\ldots,\1_{sync}(\boldc')_n)\in B_{3k}((\1_{sync}(\boldc)_i,\ldots,\1_{sync}(\boldc)_n))$. 
This implies that~$||\boldsymbol{\Delta}\cap[i,n]|-|\boldsymbol{\Delta}'\cap[i,n]||\le 3k$, where~~$\boldsymbol{\Delta}=\{i:\1_{sync}(\boldc)_i=1\}$ and~$\boldsymbol{\Delta}'=\{i:\1_{sync}(\boldc')_i=1\}$. Hence According to Eq.~\eqref{equation:difference}, we have that
\begin{align}\label{equation:upperboundg}
    |\1_{sync}( \boldc) \cdot \boldm^{(e)}- \1_{sync}(\boldc') \cdot  \boldm^{(e)}| =&|\sum^n_{i=1}(|\boldsymbol{\Delta}\cap[i,n]|-|\boldsymbol{\Delta}'\cap[i,n]|)i^e|,\nonumber\\
    \le&\sum^n_{i=1}3ki^e\nonumber\\
    <& 3kn^{e+1}.
\end{align}
If~$f(\1_{sync}(\boldc))=f(\1_{sync}(\boldc'))$, then~$\1_{sync}( \boldc) \cdot \boldm^{(e)}\equiv \1_{sync}(\boldc') \cdot  \boldm^{(e)} \bmod{3kn^{e+1}}$, which from Eq.~\eqref{equation:upperboundg} implies that~$\1_{sync}( \boldc) \cdot \boldm^{(e)}= \1_{sync}(\boldc') \cdot  \boldm^{(e)}$.  Since~$\1_{sync}(\boldc')\in B_{3k}(\1_{sync}(\boldc))$ and $\1_{sync}(\boldc),\1_{sync}(\boldc')\in\mathcal{R}_{3k}$, from Proposition \ref{lemma:vtextension} we conclude that~$\1_{sync}(\boldc)=\1_{sync}(\boldc')$. Hence Lemma~\ref{lemma:syncvector} is proved.

\subsection{Proof of Lemma~\ref{lemma:protectingpattern}}\label{subsection:proofoflemma2}
We are now ready to prove Lemma~\ref{lemma:protectingpattern}.
We have shown in Lemma~\ref{lemma:syncvector} that~$f(\1_{sync}(\boldc))\ne f(\1_{sync}(\boldc'))$ for~$\boldc'\in B_k(\boldc)\backslash \{\boldc\}$. Hence
\newline$|M(f(\1_{sync}(\boldc)))-M(f(\1_{sync}(\boldc')))|\ne 0$ (recall Eq.~\eqref{equation:Mfunction} for definition of function~$M(\boldv)$) for~$\boldc'\in B_k(\boldc)\backslash \{\boldc\}$.
According to Lemma~\ref{lemma:numberofdivisors}, the number of divisors of~$|M(f(\1_{sync}(\boldc)))-M(f(\1_{sync}(\boldc')))|$ is upper bounded by
\begin{align*}
2^{2[(3k+1)(6k+1)\ln n+(6k+1)\ln 3k]/\ln ((3k+1)(6k+1)\ln n+(6k+1)\ln 3k)}=2^{o(\log n)}. 
\end{align*}
For any sequence~$\boldc\in\{0,1\}^n$, let 
\begin{align*}
    \mathcal{P}(\boldc)=\{p: p| |M(f(\1_{sync}(\boldc')))-M(f(\1_{sync}(\boldc)))| \text{ for some $\boldc'\in B_k(\boldc)\backslash \{\boldc\}$}\} 
\end{align*}
be the set of all divisors of the numbers~$\{|M(f(\1_{sync}(\boldc')))-M(f(\1_{sync}(\boldc)))|: \boldc'\in B_k(\boldc)\backslash \{\boldc\}\}$. Since~$|B_k(\boldc)|\le \binom{n}{k}^22^{k} \le 2n^{2k}$, we have that
\begin{align*}
    |\mathcal{P}(\boldc)|\le& 2n^{2k}2^{o(\log n)}\\=&2^{2k\log n +o(\log n)}.
\end{align*}
Therefore, there exists a number~$p(\boldc)\in [1,2^{2k\log n +o(\log n)}]$
such that~$p(\boldc) \centernot | |M(f(\1_{sync}(\boldc')))-M(f(\1_{sync}(\boldc)))|$ for all~$\boldc'\in B_k(\boldc)\backslash \{\boldc\}$. 
Hence, if~$M(f(\1_{sync}(\boldc')))\equiv M(f(\1_{sync}(\boldc)))\bmod p(\boldc)$ and~$\boldc'\in B_k(\boldc)$, we have that~$p(\boldc)||M(f(\1_{sync}(\boldc')))-M(f(\1_{sync}(\boldc)))|$ and thus~$\boldc'=\boldc$. 
This completes the proof of Lemma~\ref{lemma:protectingpattern}.
\section{Hash for~\emph{$k$ dense} sequences}\label{section:kdensehash}
In this section, we present a hash function for correcting~$k$ deletions in a~$k$-\emph{dense} sequence~$\boldc$. The encoding/decoding assumes the knowledge of the \emph{synchronization vector}~$\1_{sync}(\boldc)$ and proves Lemma~\ref{lemma:recoverforkdense}. 

Let the positions of the~$1$ entries in~$\1_{sync}(\boldc)$ be ~$t_1<t_2<\ldots<t_{J}$, where~$J=\sum^n_{i=1}\1_{sync}(\boldc)_i$ is the number of~$1$ entries in ~$\1_{sync}(\boldc)$. Furthermore, let~$t_0=0$ and~$t_{J+1}=n+1$. Split~$\boldc$ into blocks~$\bolda_0,\ldots,\bolda_J$, where
\begin{align}\label{equation:aj}
\bolda_j = (c_{t_j+1},c_{t_j+2},\ldots,c_{t_{j+1}-1})
\end{align}
for~$j\in[0,J]$.
Since~$\boldc$ is~$k$-\emph{dense}, we have that the length~$|\bolda_j|$ of~$\bolda_j$ is not greater than L.
Define the hash function~$Hash_k$ as follows.
\begin{align}\label{equation:hashk}
   Hash_k(\boldc)=RS_{2k}((H(\bolda_0),\ldots,H(\bolda_J))), 
\end{align}
where~$RS_{2k}(\boldc)$ is the redundancy of a systematic Reed-Solomon code (see~\cite{Gao} for details) protecting the length~$J+1$ sequence~$(H(\bolda_0),\ldots,H(\bolda_J))$ from~$2k$ symbol substitution errors. Here the symbols are~$H(\bolda_j)$ (see Lemma~\ref{lemma:hash} for definition of~$H(\boldc)$),~$j\in [0,J]$, each having alphabet size not greater than~$2^{\lceil (L/\lceil \log n \rceil)\rceil (2k\log\log n+O(1))}$ and can be represented using~$\lceil (L/\lceil \log n \rceil)\rceil (2k\log\log n+O(1))$ bits. The length of~$Hash_k(\boldc)$ is~$\max \{4k\log (J+1),4k\lceil (L/\lceil \log n \rceil)\rceil (2k\log\log n+O(1))\}=4k\log n +o(\log n)$.
We now present the following procedure that recovers~$\boldc$ from its length~$n-k$ subsequence~$\boldd$, given the hash function~$Hash_k(\boldc)$ and the synchronization vector~$\1_{sync}(\boldc)$ recovered in Section~\ref{section:synchronizationvectors}. 
\begin{enumerate}
    \item \textbf{Step 1:} Let~$\1_{sync}(\boldd)\in\{0,1\}^{n-k}$ be the \emph{synchronization vector} of~$\boldd$.
    Recall that the locations of~$1$ entries in~$\1_{sync}(\boldc)$ are~$1\le t_1<\ldots<t_J\le n$. Let~$t_0=0$ and~$t_{J+1}=n+1$.
    \item \textbf{Step 2:} Let~$\1_{sync}(\boldd)_0=\1_{sync}(\boldd)_{n+1-k}=1$. For each~$j\in [0,J]$, if there exist two numbers~$i_j\in [t_j-k,t_j]$ and~$i_{j+1}\in [t_{j+1}-k,t_{j+1}]$ such that~$\1_{sync}(\boldd)_{i_{j}}=\1_{sync}(\boldd)_{i_{j+1}}=1$,  let~$\bolda'_j=(d_{i_j+1},d_{i_j+2},\ldots,
    d_{i_{j+1}-1})$. Else let~$\bolda'_j=0$.
    \item \textbf{Step 3:} Apply the Reed-Solomon decoder to recover~$H(\bolda_j)$ ($\bolda_j$ defined in~\eqref{equation:aj}),~$j\in [0,J]$,  from~$(H(\bolda'_0),$ $\ldots,H(\bolda'_J),$ $Hash_k(\boldc))$.
    \item \textbf{Step 4:} Let~$\boldb_j=(d_{t_j+1},d_{t_{j}+2},\ldots,d_{t_{j+1}-k-1})$, recover~$\bolda_j$ by using~$\boldb_j$ and~$H(\bolda_j)$. Then recover~$\boldc$ from~$\bolda_j$,~$j\in[0,J]$.
\end{enumerate}
Since~$\boldc_{t_j}=\1_{sync}(\boldc)_{t_j}=1$ for~$j\in[1,J]$, it suffices to show that~$\bolda_j$,~$j\in[0,J]$ can be recovered
correctly. Furthermore, note that $(d_{t_{j}+1},\ldots,d_{t_{j+1}-k-1})$ is a length~$|\bolda_j|-k$ subsequence of~$\bolda_j$, hence ~$\bolda_j$ can be correctly found given~$\boldb_j$ and~$H(\bolda_j)$,~$j\in [0.J]$.  
It is then left to recover~$H(\bolda_j)$ for~$j\in [0,J]$. To this end, we show that there are at most~$2k$ indices~$j$, such that~$\bolda'_j\ne \bolda_j$. Then there are at most~$2k$ symbol errors in the sequence~$(H(\bolda_0),\ldots,H(\bolda_J))$, which can be corrected given the Reed-Solomon code redundancy~$Hash_k(\boldc)$. 

Let~$i_j$,~$j\in[1,J]$ be the index of~$\boldc_{t_j}$ in~$\boldd$ after deletions, where~$i_j=-1$ if~$\boldc_{t_j}$ is deleted. Let~$i_0=0$ and~$i_{J+1}=n+1-k$.
The interval~$[i_j,i_{j+1}]$,~$j\in [0.J]$ is called \emph{good} if~$\1_{sync}(\boldd)_{i_j}=\1_{sync}(\boldd)_{i_{j+1}}=1$ and~$i_{j+1}-i_j=t_{j+1}-t_{j}$. Note that by definition,~$\1_{sync}(\boldd)_{i_0}=\1_{sync}(\boldd)_{i_{J+1}}=1$. Since~$\boldd_{i_j}=\boldc_{t_j}$ and~$i_j\in [t_j-k,t_j]$ for~$j\in[0,J]$, we have that~$\bolda'_{j}=\bolda_j$ if the interval~$[i_j,i_{j+1}]$ is~\emph{good}. 
Now we show that a deletion can destroy at most~$2$ \emph{good} intervals. 
Suppose that the deletion occurs in interval~$[t_j,t_{j+1}]$ in~$\boldc$ for some~$j$. 
If the deletion turns the~$1$ entry~$\1_{sync}(\boldc)_{t_j}$ to~$0$, then the~$1$ entry~$\1_{sync}(\boldc)_{t_{j+1}}$ stays. As a result, at most two \emph{good} intervals~$[i_{j-1},i_j]$ and~$[i_{j},i_{j+1}]$ are destroyed. Similarly, if the deletion turns~$\1_{sync}(\boldc)_{t_{j+1}}$ to~$0$, then at most two \emph{good} intervals~$[i_j,i_{j+1}]$ and~$[i_{j+1},i_{j+2}]$ are destroyed. If both~$1$ entries~$\1_{sync}(\boldc)_{t_j}$ and~$\1_{sync}(\boldc)_{t_{j+1}}$ are not affected, then the deletion destroys only one \emph{good} interval~$[i_{j},i_{j+1}]$. In conclusion, a deletion affects at most two~\emph{good} intervals and thus at most~$2k$ block errors~$\bolda'_j\ne\bolda_j$ occur.
Therefore, the sequence~$\boldc$ can be recovered and this proves Lemma~\ref{lemma:recoverforkdense}.

\section{Transformation to~\emph{$k$ dense} sequences}\label{section:transformation}
In this section we present an algorithm to compute the map~$T(\boldc)$ (see Lemma~\ref{lemma:transformation}), which transforms any sequence~$\boldc\in\{0,1\}^n$ into a~\emph{$k$ dense} sequence, thus proving Lemma~\ref{lemma:transformation}. 
Let~$\boldsymbol{1}^{x}$ and~$\boldsymbol{0}^y$ denote sequences of consecutive~$x$~$1$'s and consecutive~$y$~$0$'s respectively. 
We will show in Proposition~\ref{lemma:property1and2} that any sequence~$\boldc$ satisfying the following two properties is a~\emph{$k$ dense} sequence. Then, algorithms will be presented in Subsection~\ref{subsection:Tfunction} to generate sequences that satisfy the two properties and thus 
is~\emph{$k$ dense}. 
\begin{itemize}
    \item[] \textbf{Property~$1$:}
    Every length~$B\triangleq (\lceil \log k\rceil+5)2^{\lceil \log k\rceil+9}\lceil \log n \rceil$ interval of~$\boldc$ contains the pattern~$\boldsymbol{1}^{\lceil \log k\rceil+5}$, i.e., for any integer~$i\in [1,n-B+1]$, there exists an integer~$j\in [i,i+B-\lceil \log k\rceil-5]$ such that~$(c_{j},c_{j+1},\ldots,c_{j+\lceil \log k\rceil+4})= \boldsymbol{1}^{\lceil \log k\rceil+5}$.  
\item[] \textbf{property~$2$:} Every length~$R\triangleq (3k+\lceil \log k\rceil+4)(\lceil \log n\rceil +9+\lceil\log k\rceil)$ interval of~$\boldc$ contains a length~$3k+\lceil \log k\rceil+4$ subinterval that does not contain the pattern~$\boldsymbol{1}^{\lceil \log k\rceil+5}$, i.e.,   
for any integer~$i\in [1,n-R+1]$, there exists an integer~$j\in [i,i+R-3k-\lceil \log k\rceil-4]$, such that~$(c_{m},c_{m+1},\ldots,c_{m+\lceil\log k\rceil+4})\ne \boldsymbol{1}^{\lceil\log k\rceil+5}$ for every~$m\in [j,j+3k-1]$. 

\end{itemize}
\begin{proposition}\label{lemma:property1and2}
If a sequence~$\boldc$ satisfies Property~$1$ and Property~$2$, then it is a~\emph{$k$ dense} sequence.
\end{proposition}
\begin{proof}
Let 
the locations of the~$1$ entries in~$\1_{sync}(\boldc)$ be~$t_1<\ldots<t_J$. Let~$t_0=0$ and~$t_{J+1}=n+1$. It suffices to show that~$t_{i+1}-t_i\le B+R+1=L+1$ for any~$i\in [0,J]$. 

According to Property~$2$, there exists an index~$j^*\in [t_i,t_i+R-3k-\lceil \log k\rceil-4]$, such that~$(c_{m},c_{m+1},\ldots,c_{m+\lceil\log k\rceil+4})\ne \boldsymbol{1}^{\lceil\log k\rceil+5}$ for every~$m\in [j^*,j^*+3k-1]$. 
According to Property~$1$, there exists an integer~$x\in [j^*+1,j^*+B]$ such that ~$(c_{x},c_{x+1},\ldots,c_{x+\lceil\log k\rceil+4})= \boldsymbol{1}^{\lceil\log k\rceil+5}$. 
Let~$\ell=\min_{x\ge j^*,(c_{x},c_{x+1},\ldots,c_{x+\lceil\log k\rceil+4})= \boldsymbol{1}^{\lceil\log k\rceil+5}}x$. Then we have that~$\ell\in [j^*+1,j^*+B]$ and that~$(c_{m},c_{m+1},\ldots,c_{m+\lceil\log k\rceil+4})\ne \boldsymbol{1}^{\lceil\log k\rceil+5}$ for every~$m\in [j^*,\ell)$.
By definition of~$j^*$, we have that~$\ell-j^*\ge 3k$. Since~$c_{\ell},c_{\ell+1},\ldots,c_{\ell+\lceil\log k\rceil+4})= \boldsymbol{1}^{\lceil\log k\rceil+5}$, we have that~$\1_{sync}(\boldc)_{\ell}=1$. 
Therefore, we conclude that
\begin{align*}
t_{i+1}-t_i&\le \ell-t_i\\
&\le j^*+B-t_i\\
&\le R+B+1\\
&=L+1
\end{align*}
\end{proof}
The following lemma presents a function that outputs a sequence satisfying Property~$1$.
\begin{proposition}\label{lemma:frequentpattern}
For integers~$k$ and~$n>k$, there exists a map~$T_{1}:\{0,1\}^n\rightarrow \{0,1\}^{n+2\lceil \log k\rceil+10}$, computable in~$O(n^2k\log n\log^2 k)$ time, such that~$T_1(\boldc)$ satisfies Property~$1$. Moreover, either~$(T_1(\boldc)_{n+\lceil \log k\rceil+6},\ldots,T_1(\boldc)_{n+2\lceil \log k\rceil+10})=\boldsymbol{1}^{\lceil \log k\rceil+5}$ or
\newline
$(T_1(\boldc)_{n+\lceil \log k\rceil+5},\ldots,T_1(\boldc)_{n+2\lceil \log k\rceil+9})=\boldsymbol{1}^{\lceil \log k\rceil+5}$.
The sequence~$\boldc$ can be recovered from~$T_1(\boldc)$. 
\end{proposition}
\begin{proof}
We first show that
every sequence~$\boldb\in\{0,1\}^{B}$ containing no consecutive~$\lceil \log k\rceil+5$~$1$'s can be uniquely represented by a sequence~$\phi(\boldb)$ of length~$B-\lceil \log n\rceil-2\lceil \log k\rceil-12$.
Split~$\boldb$ into~$2^{\lceil \log k\rceil+9}\lceil \log n \rceil$ blocks of length~$(\lceil \log k\rceil+5)$. Since each block is not~$\boldsymbol{1}^{\lceil \log k\rceil+5}$, it can be represented by a symbol of alphabet size~$2^{\lceil \log k\rceil+5}-1$. Therefore, the sequence~$\boldb$ can be uniquely represented by a sequence~$\boldv$ of~$2^{\lceil \log k\rceil+9}\lceil \log n \rceil$ symbols, each having alphabet size~$2^{\lceil \log k\rceil+5}-1$. Convert~$\boldv$ into a binary sequence~$\phi (\boldb)$. Then $\phi(\boldb)$ can be represented by bits with length 
\begin{align}\label{equation:lengthofphib}
D\triangleq & \log_2 (2^{\lceil \log k\rceil+5}-1)^{2^{\lceil \log k\rceil+9}\lceil \log n \rceil} \nonumber\\
= &\log_2 (1-1/2^{\lceil \log k\rceil+5})^{2^{\lceil \log k\rceil+9}\lceil \log n \rceil} + (\lceil \log k\rceil+5)2^{\lceil \log k\rceil+9}\lceil \log n \rceil\nonumber\\
= &16\log_2 (1-1/2^{\lceil \log k\rceil+5})^{2^{\lceil \log k\rceil+5}} \lceil \log n \rceil+B \nonumber\\
\overset{(a)}{\le}&-16\log_2 e \lceil \log n \rceil+B\nonumber\\
\le&B-16\lceil \log n\rceil\nonumber\\
\le&B-\lceil \log n\rceil-2\lceil \log k\rceil-12,
\end{align}
where~$(a)$ follows from the fact that the function~$(1-1/x)^x$ is increasing in~$x$ for~$x>1$ and that~$\lim_{x\rightarrow \infty}(1-1/x)^x=1/e$.
Therefore,~$\phi(\boldb)$ can be represented by~$B-\lceil \log n\rceil-2\lceil \log k\rceil-12$ bits.

For a sequence~$\boldc\in\{0,1\}^n$, the encoding procedure for computing~$T_1(\boldc)$ is as follows.
\begin{enumerate}
    \item \textbf{Initialization:} Let~$T_1(\boldc)=\boldc$.
    Append~$\boldsymbol{1}^{2\lceil \log k\rceil+10}$ to the end of the sequence~$T_1(\boldc)$. Let~$i=1$ and~$n'=n$. Go to Step~$1$.
    \item \textbf{Step 1:} If~$(c_{j},c_{j+1},\ldots,c_{j+\lceil \log k\rceil+4})\ne \boldsymbol{1}^{\lceil \log k\rceil+5}$ for every~$j\in [i,i+B-\lceil \log k\rceil-5]$, go to Step 2. Else go to Step~$4$.
    \item \textbf{Step 2:} If~$i\le n'-B+1$,
    delete~$(T_1(\boldc)_{i},\ldots,$ $T_1(\boldc)_{i+B-1})$ from~$T_1(\boldc)$ and  append~$(i,\phi(T_1(\boldc)_{i},\ldots,T_1(\boldc)_{i+B-1}),0,$ $\boldsymbol{1}^{2\lceil \log k\rceil+10},0)$ to the end of~$T_1(\boldc)$. 
    Let~$n'=n'-B$ and~$i=1$. Go to Step~$1$.
    Else go to Step~$3$.
    \item \textbf{Step 3:} Delete~$(T_1(\boldc)_{i},\ldots,$ $T_1(\boldc)_{n'})$ from~$T_1(\boldc)$ and  append~$(i,\phi(T_1(\boldc)_{i},\ldots,T_1(\boldc)_{n'},\boldsymbol{0}^{i+B-n'-1}),0,$
    
    $\boldsymbol{1}^{2\lceil \log k\rceil+10-(i+B-n'-1)},0)$ to the end of~$T_1(\boldc)$. 
    Let~$n'=i-1$ and~$i=1$. Go to Step~$1$.
    \item \textbf{Step 4:}
    If~$i\le n'$, let~$i=i+1$ and go to Step~$1$.
    Else output~$T_1(\boldc)$. 
\end{enumerate}
It can be seen that the length of~$T_1(\boldc)$ keeps constant and is~$n+2\lceil \log k\rceil+10$.
The integer~$n'$ is defined such that~$(T_1(\boldc)_{n'+1},\ldots,$ $ T_1(\boldc)_{n+2\lceil \log k\rceil+10})$ are the appended bits and~$(T_1(\boldc)_{1},\ldots,T_1(\boldc)_{n'})$ are the remaining bits in~$\boldc$ after the deleting operations in Step~$2$ and Step~$3$. The appended bits are not deleted in the procedure. Note that~$(T_1(\boldc)_{n'+1},\ldots,$ $ T_1(\boldc)_{n'+\lceil \log k\rceil+5})=\boldsymbol{1}^{\lceil \log k\rceil+5}$. Hence in Step~$3$, the integer~$i$ satisfies~$1\le i+B-n'-1\le \lceil \log k\rceil+4$ and the appended bits~$\boldsymbol{1}^{2\lceil \log k\rceil+10-(i+B-n'-1)}$ has length at least~$\lceil \log k\rceil+6$.

Since either~$i$ increases from~$1$ to~$n'$ or~$n'$ decreases by~$B$ in each step.  The algorithm terminates within~$O(n^2)$ times of Step~$1$, Step~$2$ and Step~$3$. Since it takes~$O(B\log k)$ to check if a subsequence is an all~$1$'s vector. The total complexity is~$O(n^2k\log n\log^2 k)$.

We now show that the output sequence~$T_1(\boldc)$ satisfies Property~$1$.
Note that for any~$i\in [1,n']$, there exists some~$j\in [i,i+B-\lceil \log k\rceil-5]$ such that~$T_1(\boldc)_{j}=T_1(\boldc)_{j+1}=\ldots=T_1(\boldc)_{j+\lceil \log k \rceil+4}=1$. Otherwise~$T_1(\boldc)_i$ is deleted in Step~$2$ or Step~$3$. 
Moreover, the distance between two consecutive patterns~$\boldsymbol{1}^{\lceil \log k\rceil+5}$ after position~$n'$ is at most~$B-\lceil \log k\rceil-5$.   
Hence the sequence~$T_1(\boldc)$ satisfies Property~$1$. 

It can be seen that when the inserting operation in Step~$2$ or Step~$3$ does not occur, the we have that~$(T_1(\boldc)_{n+\lceil \log k\rceil+6},\ldots,$ $T_1(\boldc)_{n+2\lceil \log k\rceil+10})=\boldsymbol{1}^{\lceil \log k\rceil+5}$. If the insertion operation in Step~$2$ or Step~$3$ occurs, then we have that~$(T_1(\boldc)_{n+\lceil \log k\rceil+5},\ldots,$ $T_1(\boldc)_{n+2\lceil \log k\rceil+9})=\boldsymbol{1}^{\lceil \log k\rceil+5}$.

We now give the following decoding procedure that recovers~$\boldc$ from~$T_1(\boldc)$.
\begin{enumerate}
    \item \textbf{Initialization:} Let~$\boldc=T_1(\boldc)$ and go to Step~$1$. 
    \item \textbf{Step 1:} If~$c_{n+2\lceil \log k\rceil+10}=0$, find the length~$\ell$ of the~$1$ run that ends with~$c_{n+2\lceil \log k\rceil+9}$. let~$i$ be the decimal representation of $(c_{n+4\lceil \log k\rceil+21-B-\ell},c_{n+4\lceil \log k\rceil+22-B-\ell},\ldots,$ 
    $c_{n+4\lceil \log k\rceil+20-B-\ell+\lceil \log n\rceil})$. 
    Let~$\boldb$ be the sequence obtained by computing $\phi^{-1}(c_{n+4\lceil \log k\rceil+21-B-\ell+\lceil \log n\rceil},$ 
    $c_{n+4\lceil \log k\rceil+22-B-\ell+\lceil \log n\rceil},\ldots,c_{n+2\lceil \log k\rceil+8-\ell})$, where the function~$\phi$ is defined in the paragraph before Eq.~\eqref{equation:lengthofphib} and is invertible,
    Note that the length of~$\boldb$ is~$B$.
    Delete~$(c_{n+4\lceil \log k\rceil+21-B-\ell},$ $c_{n+4\lceil \log k\rceil+22-B-\ell},$ $\ldots,c_{n+2\lceil \log k\rceil+10})$ from~$\boldc$ and insert~$(\boldb_1,\ldots,\boldb_{B-2\lceil \log k\rceil-10+\ell})$ at location~$i$ of~$\boldc$.
    Repeat. 
    Else delete~$c_{n+1},\ldots,c_{n+2\lceil \log k\rceil+10}$ and
    output~$\boldc$.
\end{enumerate}
Note that in the encoding procedure, every appended subsequence has length~$B-2\lceil \log k\rceil-10+\ell$ and ends with a~$0$. Hence the algorithm stops when all the subsequences appended in Step~$2$ or Step~$3$ are deleted.
Moreover, the encoding procedure consists of a series of deleting and appending operations. The decoding procedure exactly reverses the series of operations in the encoding procedure. Specifically,
let~$T_{1,i}(\boldc)$,~$i\in [0,I]$ be the sequence~$T_1(\boldc)$ obtained after the~$i$-th deleting and appending operation in the encoding procedure, where~$I$ is the number of deleting and appending operations in total in the encoding procedure. We have that~$T_{1,0}(\boldc)=\boldc$ and that~$T_{1,I}(\boldc)$ is the final output~$T_{1}(\boldc)$. Then the decoding procedure obtains~$T_{1,I-i}(\boldc)$,~$i\in [0,I]$ after the~$i$-th deleting and inserting operation. Hence we get the output~$T_{1,I-I}(\boldc)=\boldc$ in the decoding procedure. 
\end{proof}
The following lemma shows that a small sequence containing~$\boldsymbol{1}^{\lceil \log k \rceil +5}$ patterns can be encoded into a sequence without the $\boldsymbol{1}^{\lceil \log k \rceil +5}$ pattern. The lemma will be used to generate sequence satisfying Property~$2$.
\begin{proposition}\label{lemma:3kto3kminus1}
For an integer~$k$,
let~$\boldc\in\{0,1\}^{3k+\lceil \log k\rceil+4}$ be a sequence 
such that $c_i=c_{i+1}=\ldots=c_{i+\lceil \log k \rceil +4}=1$ for some $i\in [1,3k]$.
There exists a mapping~$T_2:\{0,1\}^{3k+\lceil \log k\rceil+4}\rightarrow \{0,1\}^{3k+\lceil \log k\rceil+3}$, computable in~$O(k^2\log k)$ time, such that~$T_2(\boldc)$ contains no~$\lceil \log k \rceil +5$ consecutive~$1$ bits. In addition, the sequence~$\boldc$ can be recovered from~$T_2(\boldc)$.
\end{proposition}
\begin{proof}
Given~$\boldc\in\{0,1\}^{3k+\lceil \log k\rceil+4}$, the encoding procedure for computing~$T_2(\boldc)$ is given as follows.
\begin{enumerate}
    \item \textbf{Initialization:} Let~$T_2(\boldc)=\boldc$.
    Append~$0$ to the end of the sequence~$T_2(\boldc)$. Find the smallest~$i\in [1,3k]$ such that~$T_2(\boldc)_i=T_2(\boldc)_{i+1}=\ldots=T_2(\boldc)_{i+\lceil \log k \rceil +4}=1$. Delete~$(T_2(\boldc)_i,\ldots,T_2(\boldc)_{i+\lceil \log k \rceil +4})$ from~$T_2(\boldc)$ and append~$(i,\boldsymbol{0}^{\lceil \log k\rceil +3-\lceil \log (3k) \rceil})$ to the end of~$T_2(\boldc)$. Let~$n'=3k-1$ and~$i=1$.
    Go to Step~$1$.
    \item \textbf{Step 1:} If~$i\le n'-\lceil \log k \rceil-4$ and~$T_2(\boldc)_i=T_2(\boldc)_{i+1}=\ldots=T_2(\boldc)_{i+\lceil \log k \rceil +4}=1$, 
    delete~$(T_2(\boldc)_i,\ldots,T_2(\boldc)_{i+\lceil \log k \rceil +4})$ from~$T_2(\boldc)$ and  append~$(i,\boldsymbol{0}^{\lceil \log k \rceil+5-\lceil \log (3k) \rceil},1)$ to the end of~$T_2(\boldc)$. Let~$n'=n'-\lceil \log k\rceil -5$ and~$i=1$. Repeat.
    Else go to Step~$2$.
    \item \textbf{Step 2:}
    If~$i\le n'$, let~$i=i+1$ and go to Step~$1$.
    Else output~$T_2(\boldc)$. 
\end{enumerate}
The length of the sequence~$T_2(\boldc)$ keeps constant and is
\begin{align*}
    3k+\lceil \log k\rceil+4+1-\lceil \log k \rceil -5+\lceil \log k\rceil+3=3k+\lceil \log k\rceil+3.
\end{align*}
The integer~$n'$ is defined such that~$(T_2(\boldc)_{n'+1},\ldots,T_2(\boldc)_{3k+\lceil \log k\rceil+3})$ are appended bits and~$(T_2(\boldc)_{1},$ $\ldots,T_2(\boldc)_{n'})$ are the remaining bits in~$T_2(\boldc)$ after deleting~$(T_2(\boldc)_i,\ldots,T_2(\boldc)_{i+\lceil \log k \rceil +4})$. Note that~$T_2(\boldc)_{n'+1}=0$. Hence if~$(T_2(\boldc)_i,\ldots,$ $T_2(\boldc)_{i+\lceil \log k\rceil+4})=\boldsymbol{1}^{\lceil \log k\rceil+5}$, then~$i+\lceil \log k\rceil+4\le n'$ and thus~$T_2(\boldc)_{n'+1}$ is not deleted.

Since either~$i$ increases from~$1$ to~$n'$ or~$n'$ decreases in each step. The algorithm terminates within~$O(k^2)$ times of Step~$1$ and Step~$2$. Since it takes~$O(\log k)$ to check if a subsequence is an all~$1$'s vector. The total complexity is~$O(k^2\log k)$.

We now show that~$T_2(\boldc)$ contains no~$\boldsymbol{1}^{\lceil \log k\rceil+5}$ patterns. According to the encoding procedure, for~$i\in [1,n']$, we have that~$(T_2(\boldc)_i,\ldots,T_2(\boldc)_{i+\lceil \log k\rceil+4})\ne\boldsymbol{1}^{\lceil \log k\rceil+5}$. Furthermore, the distance between two consecutive~$0$ entries after position~$n'$ is at most~$\lceil \log 3k\rceil+1$. Hence for any~$i\in [n'+1,3k-1]$, we have that
~$(T_2(\boldc)_i,\ldots,T_2(\boldc)_{i+\lceil \log k\rceil+4})\ne\boldsymbol{1}^{\lceil \log k\rceil+5}$. Therefore,~$T_2(\boldc)$ does not contain~$\boldsymbol{1}^{\lceil \log k\rceil+5}$ patterns.

To show that~$T_2(\boldc)$ is decodable,
we present the following procedure that recovers~$\boldc$ from~$T_2(\boldc)$.
\begin{enumerate}
    \item \textbf{Initialization:} Let~$\boldc=T_2(\boldc)$ and go to Step~$1$. 
    \item \textbf{Step 1:} If~$c_{3k+\lceil \log k\rceil+3}=1$, let~$i$ be the decimal representation of~$(c_{3k-1},c_{3k},\ldots,c_{3k+\lceil \log 3k \rceil-2})$. 
    Delete~$(c_{3k-1},c_{3k},$ $\ldots,c_{3k+\lceil \log k\rceil+3})$ from~$\boldc$ and insert~$\boldsymbol{1}^{\lceil \log k \rceil+5}$ at location~$i$ of~$\boldc$.
    Repeat. 
    Else go to Step~$2$.
    \item \textbf{Step 2:}
    Let~$i$ be the decimal representation of~$(c_{3k+1},c_{3k+2},\ldots,c_{3k+\lceil \log (3k) \rceil})$. 
    Delete~$(c_{3k},c_{3k+2},\ldots,$ $c_{3k+\lceil \log k\rceil+3})$ from~$\boldc$ and insert~$\boldsymbol{1}^{\lceil \log k \rceil+5}$ at location~$i$ of~$\boldc$.    
    Output~$\boldc$.
\end{enumerate}
Note that in the encoding procedure, the appended subsequence in the Initialization Step ends with a~$0$. The appended subsequence in Step~$1$ ends with a~$1$. Hence the algorithm stops when~$\boldc_{3k+\lceil \log k\rceil+3}=0$ and the subsequence~$(c_{3k},c_{3k+2},\ldots,$ $c_{3k+\lceil \log k\rceil+3})$ is deleted.

Similarly to the proof of correctness of decoding in Proposition~\ref{lemma:frequentpattern}, the decoding procedure exactly reverses the series of operations in the encoding procedure.
Let $T_{2,i}(\boldc)$,~$i\in [0,I]$ be the sequence obtained after the~$i$-th deleting and appending operation in the encoding procedure, where
$I$ is the number of deleting and appending operations in total in the encoding procedure.
Then~$T_{2,i}(\boldc)$ is the sequence obtained after the~$I-i$-th deleting and inserting operation in the decoding procedure. Therefore, the decoding procedure recovers the sequence~$\boldc$ after the~$I$-th operation. 
\end{proof}
\subsection{Proof of Lemma~\ref{lemma:transformation}}\label{subsection:Tfunction}
We are now ready to present the encoding and decoding procedure for computing~$T(\boldc)$.
The encoding procedure is as follows.
\begin{enumerate}
    \item \textbf{Initialization:} Let~$T(\boldc)=T_1(\boldc)$.
    Append~$(\boldsymbol{0}^{3k},\boldsymbol{1}^{\lceil \log k \rceil +5})$ to the end of the sequence~$T(\boldc)$. Let~$n'=n+2\lceil\log k\rceil+10$ (the length of~$T_1(\boldc)$) and~$i=1$.
    Go to Step~$1$.
    \item \textbf{Step 1:} If~$i\le \min\{n',n+3k+3\lceil\log k\rceil+16-R\}$ and
    for every~$j\in [i,i+R-3k-\lceil \log k\rceil-4]$, there exists~$m\in [j,j+3k-1]$ such that~$(T(\boldc)_m,T(\boldc)_{m+1},\ldots,$ $T(\boldc)_{m+\lceil \log k \rceil +4})=\boldsymbol{1}^{\lceil \log k \rceil +5}$, split~$(T(\boldc)_i,T(\boldc)_{i+1},\ldots,T(\boldc)_{i+R-1})$ into~$(\lceil \log n\rceil +9+\lceil\log k\rceil)$ blocks~$\boldb_1,\boldb_2,\ldots,\boldb_{\lceil \log n\rceil +9+\lceil\log k\rceil}$ of length~$3k+\lceil \log k \rceil +4$. Delete~$(\boldb_2,\ldots,\boldb_{\lceil \log n\rceil +8+\lceil\log k\rceil})$ from~$T(\boldc)$ and append~$(0,T_2(\boldb_2),T_2(\boldb_3),\ldots,T_2(\boldb_{\lceil \log n\rceil+8+\lceil\log k\rceil}),i+3k+\lceil\log k\rceil+4,\boldsymbol{1}^{\lceil\log k\rceil+5},0)$ to the end of
    ~$T(\boldc)$. Let~$n'=n'-R+6k+2\lceil\log k\rceil+8$ and~$i=1$. Repeat.
    Else go to Step~$2$.
    \item \textbf{Step 2:}
    If~$i\le n'$, let~$i=i+1$ and go to Step~$1$.
    Else output~$T(\boldc)$. 
\end{enumerate}
The length of~$T(\boldc)$ keeps constant and is~$n+3k+3\lceil\log k\rceil+15$. 
The subsequence~$T(\boldc)_{n'+1},\ldots,$ $T(\boldc)_{n+3k+3\lceil\log k\rceil+15}$ are appended bits. 
The subsequence~$T(\boldc)_1,\ldots,T(\boldc)_{n'}$ consists of the remaining bits in~$T(\boldc)$ after deletions in Step~$1$. Note that~$(T(\boldc)_{n'+1},\ldots,T(\boldc)_{n+3k+3\lceil\log k\rceil+15})$ is not deleted.

We show that~$T(\boldc)_{n'+1},\ldots,T(\boldc)_{n'+3k}=\boldsymbol{0}^{3k}$ and that either~$T(\boldc)_{n'-\lceil\log k\rceil-4},\ldots,$ $T(\boldc)_{n'}=\boldsymbol{1}^{\lceil\log k\rceil+5}$ or~$T(\boldc)_{n'-\lceil\log k\rceil-5},\ldots,$ $T(\boldc)_{n'-1}=\boldsymbol{1}^{\lceil\log k\rceil+5}$. 
The proof is based on induction on the number~$r$ of deleting operations in Step~$1$ that have been done.
According to Proposition~\ref{lemma:frequentpattern}, the claim holds after the Initialization Step and hence after~$r=0$ deleting operations. Suppose the claim holds after~$r$ deleting operations and we have that~$(T(\boldc)_{n'+1},\ldots,T(\boldc)_{n'+3k})=\boldsymbol{0}^{3k}$. Then in the~$r+1$-deleting operation in Step~$1$, the "if" condition only holds when~$i\le n'-R$. Otherwise, for~$j=n'\in [i,i+R-1]$, we have that~$(c_m,c_{m+1},\ldots,$ $c_{m+\lceil \log k \rceil +4})\ne\boldsymbol{1}^{\lceil \log k \rceil +5}$ for every~$m\in [j,j+3k-1]$, contradicting to the "if" condition. Furthermore, note that the bits~$(T(\boldc)_{i+R-3k-\lceil\log k\rceil-4},\ldots,T(\boldc)_{n+3k+3\lceil\log k\rceil+15})$ after block~$\boldb_{\lceil \log n\rceil+9+\lceil\log k\rceil}$ are not deleted in the deleting operation.
Hence the bits~$T(\boldc)_{n'-\lceil\log k\rceil-5},\ldots,T(\boldc)_{n+3k+3\lceil\log k\rceil+15}$ stay in~$T(\boldc)$. Therefore, the claim holds after~$r+1$-th deleting operation and holds after any number of deleting operations. 

We now show that~$T(\boldc)$ satisfies Property~$1$.
From Proposition~\ref{lemma:frequentpattern}, the initial sequence~$T(\boldc)=(T_1(\boldc),\boldsymbol{0}^{3k},\boldsymbol{1}^{\lceil\log k\rceil+5})$ satisfies Property~$1$. 
We prove that~$T(\boldc)$ keeps Property~$1$ after the deleting and inserting operations. 
Note that the deleting operation does not delete~$\boldb_1$ and~$\boldb_{\lceil \log n\rceil +9+\lceil\log k\rceil}$, which both contain~$\boldsymbol{1}^{\lceil\log k\rceil+5}$ as a subsequence. Hence the deleting operation do not change the distance of consecutive~$\boldsymbol{1}^{\lceil\log k\rceil+5}$ patterns before block~$\boldb_1$ and after~$\boldb_{\lceil \log n\rceil +9+\lceil\log k\rceil}$. Since the distance of consecutive~$\boldsymbol{1}^{\lceil\log k\rceil+5}$ patterns in~$(\boldb_1,\boldb_{\lceil \log n\rceil +9+\lceil\log k\rceil})$ is at most~$6k+2\lceil\log k\rceil+8\le B-\lceil\log k\rceil-5$, the sequence~$T(\boldc)$ keeps Property~$1$ after a deleting operation. Moreover, the distance of consecutive~$\boldsymbol{1}^{\lceil\log k\rceil+5}$ patterns in the
appended bits~$(T(\boldc)_{n'+1},\ldots,T(\boldc)_{n+3k+3\lceil\log k\rceil+15})$ is at most~$R-6k-2\lceil\log k\rceil-8\le B-\lceil\log k\rceil-5$. Note that the appended bits are not deleted in the deleting operations. Hence~$T(\boldc)$ keeps Property~$1$ after an inserting operation and we conclude that the output~$T(\boldc)$ satisfies Property~$1$. 

Next, we prove that~$T(\boldc)$ satisfies Property~$2$. According to the encoding procedure, for any~$i\in [1,\min\{n',n+3k+3\lceil\log k\rceil+16-R\}]$, there exists some~$j\in[i,i+R-3k-\lceil \log k\rceil-4]$, such that~$(T(\boldc)_m,T(\boldc)_{m+1},\ldots,$ $T(\boldc)_{m+\lceil \log k \rceil +4})\ne\boldsymbol{1}^{\lceil \log k \rceil +5}$ for every~$m\in [j,j+3k-1]$.  
Note that the appended bits~$(T(\boldc)_{n'+1},\ldots,T(\boldc)_{n+3k+3\lceil\log k\rceil+15})$ are not deleted. Hence for~$i\in [n'+1,n+3k+3\lceil\log k\rceil+16-R]$, the interval~$[i,i+R-1]$ contains a subinterval~$[j,j+3k+\lceil\log k\rceil+3]$ such that~$T(\boldc)_{j}=0$ and~$(T(\boldc)_{j+1},\ldots,T(\boldc)_{j+3k+\lceil\log k\rceil+3})=T_2(\boldb_2)$, where~$\boldb_2\in\{0,1\}^{3k+\lceil\log k\rceil+4}$ contains the~$\boldsymbol{1}^{\lceil\log k\rceil+5}$ pattern. Then~$(T(\boldc)_m,$ $T(\boldc)_{m+1},\ldots,$ $T(\boldc)_{m+\lceil \log k \rceil +4})\ne\boldsymbol{1}^{\lceil \log k \rceil +5}$ for every~$m\in [j,j+3k-1]$. Therefore, the sequence~$T(\boldc)$ satisfies Property~$2$. According to Proposition~\ref{lemma:property1and2}, we conclude that~$T(\boldc)$ satisfies Property~$1$ and Property~$2$ and is~\emph{$k$ dense}.

Since either~$i$ increases from~$1$ to~$n'$ or~$n'$ decreases in the encoding procedure, the procedure terminates within~$O(n^2)$ iterations. Each iteration takes~$O(k\log kR)$ time to check a violation of Property~$2$. Hence the complexity is at most~$O(n^2k^2\log k(\log n))$. Therefore, the total complexity is~$poly(n,k)$.

Finally we present the following decoding procedure that recovers~$\boldc$ from~$T(\boldc)$.
\begin{enumerate}
    \item \textbf{Initialization:} Let~$\boldc=T(\boldc)$ and go to Step~$1$. 
    \item \textbf{Step 1:} If~$c_{n+3k+3\lceil\log k\rceil+15}=0$, let~$i$ be the decimal representation of~$(c_{n+3k+2\lceil\log k\rceil+10-\lceil \log n\rceil},$
    \newline
    $c_{n+3k+2\lceil\log k\rceil+11-\lceil \log n\rceil},\ldots,c_{n+3k+2\lceil\log k\rceil+9})$. Split~$(c_{n+9k+5\lceil\log k\rceil+25-R},\ldots,c_{n+3k+2\lceil\log k\rceil+9-\lceil \log n\rceil})$ into 
    \newline
    $\lceil \log n\rceil+7+\lceil\log k\rceil$ blocks~$(\boldb'_1,\ldots,\boldb'_{\lceil \log n\rceil+7+\lceil\log k\rceil})$ of length~$3k+\lceil\log k\rceil+3$. Compute~$\boldb_j=T^{-1}_2(\boldb'_j)$ for~$j\in [1,\lceil \log n\rceil+7+\lceil\log k\rceil]$, where~$T^{-1}_2(\boldb'_j)$ is obtained by applying~$T_2$ decoder (Proposition~\ref{lemma:3kto3kminus1}) on~$\boldb'_j$.
    Delete \newline 
    $(c_{n+9k+5\lceil\log k\rceil+24-R},\ldots,$ $c_{n+3k+3\lceil\log k\rceil+15})$ from~$\boldc$ and insert~$\boldb_1,\ldots,\boldb_{\lceil \log n\rceil+7+\lceil\log k\rceil}$ at location~$i$ of~$\boldc$.
    Repeat. 
    Else delete~$(c_{n+2\lceil\log k\rceil+11},\ldots,c_{n+3k+3\lceil\log k\rceil+15})$ and output~$T^{-1}_1(\boldc)$.
\end{enumerate}
According to the encoding procedure, the inserted bits end with a~$1$ entry in the Initialization Step and with a~$0$ entry in Step~$1$. Note that the inserted bits are not deleted in the encoding procedure. Hence the decoding algorithm stops when an ending~$1$ entry is detected.

Similarly to the 
to the proof of correctness of decoding in Proposition~\ref{lemma:frequentpattern} and Proposition~\ref{lemma:3kto3kminus1}, the decoding procedure exactly reverses the series of operations in the encoding procedure. Therefore, the decoding procedure decodes the sequence~$\boldc$ correctly. The proof is done.
\section{Encoding}\label{section:encoding}
In this section we present the encoding function~$\mathcal{E}$ and prove Theorem~\ref{theorem:main}.
The function~$\mathcal{E}$ is given by
\begin{align*}
    \mathcal{E}(\boldc) = (T(\boldc),R'(\boldc),R''(\boldc)).
\end{align*}
where 
\begin{align*}
   R'(\boldc)&=( M(f(\1_{sync}(T(\boldc)))) \bmod p(T(\boldc)),p(T(\boldc)),Hash_k(T(\boldc))), \text{ and }\\
   R''(\boldc)&=Rep_{k+1}(H(R'(\boldc))).
\end{align*}
Here~$M(\boldv)$ is the function defined in Eq.~\eqref{equation:Mfunction} and~$Rep_{k+1}(H(R'(\boldc)))$ is the~$k+1$-fold repetition of the bits in~$H(R'(\boldc))$ (See Lemma~\ref{lemma:hash} for definition of function~$H(\boldc)$). The function~$Hash_k(\boldc)$ is defined in Eq.~\eqref{equation:hashk} and the function~$T(\boldc)$ is defined in Lemma~\ref{lemma:transformation}.

The length of~$R'(\boldc)$ is~$N_1=8k\log (n+3k+3\lceil \log k\rceil +15)+o(\log (n+3k+3\lceil \log k\rceil +15))=8k\log n +o(\log n)$. The length of~$R''(\boldc)$ is~$N_2=2k(k+1)(N_1/\lceil \log n\rceil )\log \log n=o(\log n)$. 
Therefore, the redundancy of~$\mathcal{E}(\boldc)$ is~$N_1+N_2=8k\log n+o(\log n)$. 
Let~$N=n+8k\log n +o(\log n)$ be the length of~$\mathcal{E}(\boldc)$. To show that~$\boldc$ can be recovered from a length~$N-k$ subsequence~$\boldd$ of~$\mathcal{E}(\boldc)$, we prove that
\begin{enumerate}
    \item The redundancy~$R'(\boldc)$ can be recovered from~$k$ deletions with the help of~$R''(\boldc)$.
    \item The sequence~$T(\boldc)$ can be recovered from~$k$ deletions with the help of~$R'(\boldc)$.
\end{enumerate}
Note that~$(d_{n+N_1+1},\ldots,d_{n+N_1+N_2-k})$ is a length~$N_2-k$ subsequence of~$R''(\boldc)$. Since~$R''(\boldc)$ is a~$k$ deletion code protecting~$H(R'(\boldc))$, the hash function $H(R'(\boldc))$ can be recovered from~$(d_{n+N_1+1},\ldots,d_{n+N_1+N_2-k})$.
Moreover,~$(d_{n+1},\ldots,$ $d_{n+N_1-k})$ is a length~$N_1-k$ subsequence of~$R'(\boldc)$. From Lemma~\ref{lemma:hash} the hash function~$H(R'(\boldc))$ protects~$R'(\boldc)$ from~$k$ deletions. Hence~$(1)$ holds. 

We now prove~$(2)$. According to Lemma~\ref{lemma:protectingpattern}, the \emph{synchronization vector}~$\1_{sync}(T(\boldc))$ can be recovered from~$M(f(\1_{sync}(T(\boldc)))) $ $\bmod p(T(\boldc))$ and~$p(T(\boldc))$. Since by Lemma~\ref{lemma:transformation},~$T(\boldc)$ is a~\emph{$k$ dense} sequence. Hence from Lemma~\ref{lemma:recoverforkdense}, it can be recovered based on~$\1_{sync}(T(\boldc))$ and~$Hash_k(T(\boldc))$. Finally, the sequence~$\boldc$
can be recovered from~$T(\boldc)$ by Lemma~\ref{lemma:transformation}. Hence~$(2)$ holds and~$\boldc$ can be recovered.

The encoding complexity of~$\mathcal{E}(\boldc)$ is~$O(n^{2k+1})$, which comes from brute fore searching for ~$p(T(\boldc))$. The decoding compliexity is~$O(n^k+1)$, which comes from brute force searching for the correct~$\1_{sync}(T(\boldc))$, given~$M(f(\1_{sync}(T(\boldc)))) \bmod p(T(\boldc))$ and~$p(T(\boldc))$.

\section{Conclusion and Future Work}\label{section:conclusion}
We construct a~$k$-deletion correcting code with optimal order redundancy.  
Interesting open problems include finding complexity~$O(N^{O(1)})$ encoding/decoding algorithms for our code, as well as constructing a systematic $k$-deletion code with optimal redundancy.

\bibliographystyle{IEEEtran}

\begin{thebibliography}{1}
\bibitem{levenshtein1966binary}
V.~I. Levenshtein, ``Binary codes capable of correcting deletions, insertions, and reversals,'' \emph{Soviet physics doklady}, vol.~10, no.~8, pp.~707--710, 1966.

\bibitem{vt1965}
R.~R. Varshamov and G.~M. Tenengolts, ``Codes which correct single asymmetric errors,'' \emph{Autom. Remote Control}, vol.~26, no.~2, pp.~286--290, 1965.

\bibitem{Mitzenmacher}
M.~Mitzenmacher, ``A survey of results for deletion channels and related synchronization channels,'' \emph{Probability Surveys}, vol.~6, pp.~1--33, 2009.

\bibitem{helberg2002multiple}
A.~S. Helberg and H.~C. Ferreira, ``On multiple insertion/deletion correcting codes,'' \emph{IEEE Trans. on Inf. Th.}, vol.~48, no.~1, pp.~305--308, 2002.

\bibitem{abdel2012helberg}
K.~A. Abdel-Ghaffar, F.~Paluncic, H.~C. Ferreira, and W.~A. Clarke, ``On Helberg's generalization of the Levenshtein code for multiple deletion/insertion error correction,'' \emph{IEEE Trans. on Inf. Th.}, vol.~58, no.~3, pp.~1804--1808, 2012.



\bibitem{paluncic2012multiple}
F.~Paluncic, K.~A. Abdel-Ghaffar, H.~C. Ferreira, and W.~A. Clarke, ``A multiple insertion/deletion correcting code for run-length limited sequences,'' \emph{IEEE Trans. on Inf. Th.}, vol.~58, no.~3, pp.~1809--1824, 2012.



\bibitem{brakensiek2016efficient}
J.~Brakensiek, V.~Guruswami, and S.~Zbarsky, ``Efficient low-redundancy codes for correcting multiple deletions,'' \emph{IEEE Trans. on Inf. Th.}, vol.~64, no.~5, pp.~3403--3410, 2018.

\bibitem{Haeupler}
B.~Haeupler, ``Optimal document exchange and new codes for small number of insertions and deletions.''  \emph{arXiv:1804.03604} [cs.DS], 2018.

\bibitem{kuan2018deterministic}
K.~Cheng, Z.~Jin, X.~Li and K.~Wu, ``Deterministic document exchange protocols, and almost optimal binary codes for edit errors,'' \emph{IEEE 59th Annual Symposium on Foundations of Computer Science (FOCS)}, pp.~200--211, 2018.


\bibitem{schulman1999asymptotically}
L.~J.~Schulman and D.~Zuckerman, ``Asymptotically good codes correcting insertions, deletions, and transpositions,'' \emph{IEEE Trans. on Inf. Th.}, vol.~45, no.~7, pp.~2552--2557, 1999.

\bibitem{guruswami2017deletion}
V.~Guruswami and C.~Wang, ``Deletion codes in the high-noise and high-rate regimes,'' \emph{IEEE Trans. on Inf. Th.}, vol.~63, no.~4, pp.~1961--1970, 2017.

\bibitem{Ryan}
R.~Gabrys and F.~Sala, ``Codes correcting two deletions,'' \emph{IEEE Trans. on Inf. Th.}, vol. 65, no. 2, pp. 965--974, Feb. 2019.

\bibitem{sima2018two}
J.~Sima, N.~Raviv, and J.~Bruck, ``Two Deletion Correcting Codes from Indicator Vectors,'' \emph{IEEE Int. Symp. on Inform. Theory.}, Vail, USA, pp.~421--425, 2018.




\bibitem{hanna2018guess}
S.~K. Hanna and S.~El Rouayheb, ``Guess \& check codes for deletions, insertions, and synchronization,'' \emph{IEEE Trans. on Inf. Th.}, vol. 65, no. 1, pp.~3--15, Jan. 2019.



\bibitem{Nicolas}
J.~L.~Nicolas, ``On highly composite numbers,'' in \emph{Ramanujan revisited}, Urbana-Champaign, Ill., pp.~215-244, 1987.

\bibitem{Gao}
S.~Gao, "A new algorithm for decoding Reed-Solomon codes," \emph{Communications, Information and Network Security}, Springer, Boston, MA,  pp.~55--68, 2003.
\end{thebibliography}

\end{document}